\pdfoutput=1

\documentclass[11pt,twoside,a4paper,notitlepage]{article}
\usepackage[a4paper,margin=2.5cm]{geometry}

\newcommand{\ShortVersionOnly}[1]{}

\newcommand{\ProofIsInAppendix}[1]{}
\newcommand{\ProofIsInFullVersion}[1]{}

\title{%
Constant delay enumeration with FPT-preprocessing for conjunctive
queries of bounded submodular width\footnote{This is the full version
  of the conference contribution \cite{DBLP:conf/mfcs/BerkholzS19}.}
}%

\usepackage{ifthen}

\usepackage[utf8]{inputenc}

\usepackage{color}

\usepackage{tikz}
\usepackage[noend]{algpseudocode}
\usepackage{xspace}

\usepackage{bm}

\usepackage{algorithm}

\usepackage{microtype}%

\usepackage{amsmath} 
\usepackage{amsthm}
\usepackage{amsfonts} 
\usepackage{amssymb}   
\usepackage{stmaryrd} %
\usepackage{enumerate}

  \author{Christoph Berkholz, Nicole Schweikardt \\
    Humboldt-Universität zu Berlin \\
    \texttt{\{berkholz,schweikn\}@informatik.hu-berlin.de}
  }

  \usepackage{hyperref}

       \newtheorem{theorem}{Theorem}[section] 
        \newtheorem{lemma}[theorem]{Lemma}
        \newtheorem{corollary}[theorem]{Corollary}

        \theoremstyle{definition}
        \newtheorem{definition}[theorem]{Definition}

\newcommand{\nc}[1]{\newcommand{#1}}
\newcommand{\rnc}[1]{\renewcommand{#1}}

\nc{\myparagraph}[1]{\textbf{#1.}}

\rnc{\leq}{\ensuremath{\leqslant}}
\rnc{\geq}{\ensuremath{\geqslant}}

\rnc{\le}{\leq}
\rnc{\ge}{\geq}

\nc{\isdef}{\ensuremath{:=}}
\nc{\deff}{\isdef}
\nc{\defi}{\isdef}

\nc{\set}[1]{\ensuremath{\{#1\}}}
\nc{\setsize}[1]{\ensuremath{|#1|}}
\nc{\Setsize}[1]{\ensuremath{\big|#1\big|}}
\nc{\Set}[1]{\ensuremath{\big\{#1\big\}}}
\nc{\setc}[2]{\set{#1 \ : \ #2}}
\nc{\Setc}[2]{\Set{#1 \ : \ #2}}

\nc{\aufgerundet}[1]{\ensuremath{\lceil #1 \rceil}}
\nc{\abgerundet}[1]{\ensuremath{\lfloor #1 \rfloor}}

\nc{\dcup}{\ensuremath{\dot\cup}}

\nc{\ov}[1]{\ensuremath{\overline{#1}}}

\nc{\NN}{\ensuremath{\mathbb{N}}}
\nc{\NNpos}{\ensuremath{\NN_{\scriptscriptstyle\geq 1}}}
\nc{\RR}{\ensuremath{\mathbb{R}}}
\nc{\RRpos}{\ensuremath{\RR_{\scriptscriptstyle\geq 0}}}
\nc{\QQ}{\ensuremath{\mathbb{Q}}}
\nc{\QQpos}{\ensuremath{\QQ_{\scriptscriptstyle\geq 0}}}

\nc{\und}{\ensuremath{\wedge}}
\nc{\Und}{\ensuremath{\bigwedge}}
\nc{\oder}{\ensuremath{\vee}}
\nc{\Oder}{\ensuremath{\bigvee}}
\nc{\nicht}{\ensuremath{\neg}}
\nc{\impl}{\ensuremath{\to}}
\nc{\gdw}{\ensuremath{\leftrightarrow}}

\newcommand{\uund}{\,\und\,}

\nc{\free}{\ensuremath{\textrm{\upshape free}}}
\nc{\quant}{\ensuremath{\textrm{\upshape quant}}}
\nc{\ar}{\ensuremath{\operatorname{ar}}}

\nc{\Structure}[1]{\ensuremath{\mathcal{#1}}}
\nc{\A}{\Structure{A}}
\nc{\B}{\Structure{B}}
\nc{\C}{\Structure{C}}

\nc{\isom}{\ensuremath{\cong}}

\nc{\querycont}{\ensuremath{\sqsubseteq}}
\nc{\eval}[2]{\ensuremath{\llbracket#1\rrbracket^{#2}}}
\nc{\semantik}[1]{\ensuremath{\left\llbracket#1\right\rrbracket}}

\nc{\CanDB}[1]{\ensuremath{\A_{#1}}} %
\nc{\CanTup}[1]{\ensuremath{t_{#1}}} %

\newcommand{\queryphi}{\varphi}
\newcommand{\varv}{v}

\newcommand{\relS}{S} %
 
\newcommand{\relT}{T} %

\newcommand{\relE}{E} %

\renewcommand{\epsilon}{\varepsilon}

\newcommand{\strucA}{\mathcal A}

\newcommand{\strucB}{\mathcal B}
\newcommand{\strucC}{\mathcal C}
\nc{\Vars}{\ensuremath{\textrm{\upshape vars}}}
\nc{\vars}{\Vars}
\nc{\Cons}{\ensuremath{\textrm{\upshape cons}}}
\nc{\cons}{\Cons}
\nc{\atoms}{\ensuremath{\textrm{\upshape atoms}}}
\nc{\Atoms}{\atoms}
\nc{\Adom}{\ensuremath{\textrm{\upshape adom}}}
\nc{\adom}[1]{\ensuremath{\Adom(#1)}} %
\nc{\dom}[1]{\ensuremath{\textrm{\upshape dom}(#1)}} %

\newcommand{\DBone}[1]{}

\newcommand{\bigoh}{O}
\newcommand{\bigOh}{\bigoh}

\newcommand{\parent}{\pointerfont{parent}}

\nc{\arrayfont}[1]{\ensuremath{\texttt{#1}}}

\newcommand{\size}[1]{\ensuremath{|\!|#1|\!|}}
\nc{\card}[1]{\ensuremath{|#1|}}

\newcommand{\assign}{\ensuremath{\alpha}}

\newcommand{\potenzmengeof}[1]{2^{#1}}

\nc{\insertp}{\textsc{Insert}}
\nc{\cleanup}{\textsc{cleanUp}}
\nc{\cleanups}{\textsc{cleanUp'}}

\nc{\Yes}{\texttt{yes}}
\nc{\No}{\texttt{no}}

\nc{\Dom}{\ensuremath{\textbf{dom}}}
\nc{\Var}{\ensuremath{\textbf{var}}}
\nc{\schema}{\ensuremath{\sigma}}
\nc{\DB}{\ensuremath{D}} %
\nc{\DBstrich}{\ensuremath{D'}} %
\nc{\DBstart}{\ensuremath{{\DB_0}}} %
\nc{\DBempty}{\ensuremath{{\DB_{\emptyset}}}} %

\nc{\DS}{\ensuremath{\mathtt{D}}} %

\rnc{\phi}{\queryphi}

\nc{\UpdateFont}[1]{\ensuremath{\textsf{#1}}}
\nc{\Delete}{\UpdateFont{delete}}
\nc{\Insert}{\UpdateFont{insert}}
\nc{\Update}{\UpdateFont{update}}

\nc{\AlgoFont}[1]{\ensuremath{\textbf{#1}}}
\nc{\PREPROCESS}{\AlgoFont{preprocess}}
\nc{\INIT}{\AlgoFont{init}}
\nc{\UPDATE}{\AlgoFont{update}}
\nc{\ENUMERATE}{\AlgoFont{enumerate}}
\nc{\COUNT}{\AlgoFont{count}}
\nc{\ANSWER}{\AlgoFont{answer}}
\nc{\TEST}{\AlgoFont{test}}

\nc{\EOE}{\texttt{\upshape EOE}\xspace} %

\nc{\preprocessingtime}{\ensuremath{t_p}}
\nc{\inittime}{\ensuremath{t_i}}
\nc{\delaytime}{\ensuremath{t_d}}
\nc{\updatetime}{\ensuremath{t_u}}
\nc{\answertime}{\ensuremath{t_a}}
\nc{\countingtime}{\ensuremath{t_c}}
\nc{\testingtime}{\ensuremath{t_t}}

\nc{\preprocessingtimehat}{\ensuremath{\hat{t}_p}}
\nc{\inittimehat}{\ensuremath{\hat{t}_i}}
\nc{\delaytimehat}{\ensuremath{\hat{t}_d}}
\nc{\updatetimehat}{\ensuremath{\hat{t}_u}}
\nc{\answertimehat}{\ensuremath{\hat{t}_a}}
\nc{\countingtimehat}{\ensuremath{\hat{t}_c}}
\nc{\testingtimehat}{\ensuremath{\hat{t}_t}}

\nc{\phiBTypical}{\ensuremath{\phi'_{\relS\text{-}\relE\text{-}\relT}}}
\nc{\phiJTypical}{\ensuremath{\phi_{\relS\text{-}\relE\text{-}\relT}}}
\nc{\phiET}{\ensuremath{\phi_{\relE\text{-}\relT}}}

\nc{\restrict}[2]{\ensuremath{{#1}_{|#2}}}
\nc{\extend}[3]{\ensuremath{{#1}\tfrac{#3}{#2}}}
\nc{\valuation}{\ensuremath{\beta}}
\nc{\emptyassign}{\ensuremath{\emptyset}}
\nc{\Assign}[2]{\ensuremath{\frac{#2}{#1}}}

\nc{\vroot}{\ensuremath{\varv_{\textsl{root}}}}

\nc{\pointerfont}[1]{\textit{#1}}

\nc{\varitem}[1]{\ensuremath{v^{#1}}}
\nc{\assitem}[1]{\ensuremath{\assign^{#1}}}
\nc{\constitem}[1]{\ensuremath{a^{#1}}}
\nc{\parentitem}[1]{\ensuremath{\parent^{#1}}}
\nc{\childitem}[2]{\ensuremath{\pointerfont{child}^{#1}_{#2}}}
\nc{\llist}[2]{\ensuremath{\mathcal{L}_{#2}^{#1}}}
\nc{\startlist}{\ensuremath{\mathcal{L}_{\text{\upshape start}}}\xspace}
\nc{\nextlistitem}[1]{\ensuremath{\pointerfont{next-listitem}^{#1}}}
\nc{\prevlistitem}[1]{\ensuremath{\pointerfont{prev-listitem}^{#1}}}
\nc{\countitem}[1]{\ensuremath{C_{\textit{below}}^{#1}}}
\nc{\desc}[1]{\ensuremath{\text{desc}}}

\nc{\Null}{\ensuremath{0}}

\nc{\arrayA}{\arrayfont{A}}
\nc{\arrayB}{\arrayfont{B}}
\nc{\arrayC}{\arrayfont{C}}
\nc{\arrayE}{\arrayfont{E}}

\nc{\ITEMS}{\mathcal{I}}

\nc{\NIL}{\textsc{nil}}

\nc{\TupleSet}{\ensuremath{\mathcal{T}}}
\nc{\ResultSet}{\ensuremath{\mathcal{R}}}
\nc{\SkipArrayNext}[1]{\ensuremath{\mathsf{skip}[#1].\mathsf{next}}}
\nc{\SkipArrayPrev}[1]{\ensuremath{\mathsf{skip}[#1].\mathsf{prev}}}
\nc{\AlgoA}{\ensuremath{\mathbb{A}}}
\nc{\nil}{\texttt{nil}\xspace}
\nc{\SkipStart}{\ensuremath{\mathsf{sk{-}start}}}
\nc{\tup}{\ensuremath{\ov{t}}}
\nc{\tups}{\ensuremath{\ov{s}}}
\nc{\prozvisit}{\ensuremath{\textsc{Visit}}}
\nc{\prozvisitrev}{\ensuremath{\textsc{Visit}^{-1}}}
\nc{\tut}{\ensuremath{t}}
\nc{\enumprev}{\ensuremath{\vartriangleleft}}
\nc{\SkipLast}{\ensuremath{\mathsf{sk{-}last}}}

\nc{\lllist}{\ensuremath{\mathcal{L}}}
\nc{\pllist}{\ensuremath{\mathcal{L}^+}}
\nc{\milist}{\ensuremath{\mathcal{L}^-}}
\nc{\cilist}{\ensuremath{\mathcal{L}^\circ}}
\nc{\numitmpl}{\ensuremath{+{-}\text{on}{-}\text{path}}}
\nc{\numitmmi}{\ensuremath{-{-}\text{on}{-}\text{path}}}
\nc{\numitmci}{\ensuremath{\circ{-}\text{on}{-}\text{path}}}
\nc{\DBnew}{\ensuremath{\DB_{\text{new}}}}
\nc{\DBold}{\ensuremath{\DB_{\text{old}}}}
\nc{\liitmpl}{\ensuremath{\mathcal{L}^{+{-}\text{on}{-}\text{path}}}}
\nc{\liitmmi}{\ensuremath{\mathcal{L}^{-{-}\text{on}{-}\text{path}}}}
\nc{\liitmci}{\ensuremath{\mathcal{L}^{\circ{-}\text{on}{-}\text{path}}}}
\nc{\ITEMSres}[1]{\ensuremath{\ITEMS|_{#1}}}

\nc{\prVisit}{\textsc{Visit}}
\nc{\prVisitRes}{\textsc{VisitRes}}
\nc{\prEnumWithItem}{\textsc{EnumWithItem}}
\nc{\prFindItems}{\textsc{FindItems}}

\nc{\emptytuple}{\ensuremath{()}}
\nc{\emptyword}{\ensuremath{\varepsilon}}

\nc{\proj}{\ensuremath{\pi}}
\nc{\select}{\ensuremath{\sigma}}

\nc{\FD}{\ensuremath{\delta_{\textit{fd}}}} %
\nc{\IND}{\ensuremath{\delta_{\textit{ind}}}} %
\nc{\INDtilde}{\ensuremath{\tilde{\delta}_{\textit{ind}}}}%
\nc{\SD}{\ensuremath{\delta_{\textit{sd}}}} %
\nc{\CC}{\ensuremath{\delta_{\textit{cc}}}}%

\nc{\DEP}{\ensuremath{\delta}} %
\nc{\CONSTR}{\ensuremath{\Sigma}} %

\nc{\qSET}{\ensuremath{q_{\textit{S-E-T}}}}
\nc{\pSET}{\ensuremath{p_{\textit{S-E-T}}}}

\nc{\qET}{\ensuremath{q_{\textit{E-T}}}}

\nc{\Ans}{\ensuremath{\textit{Ans}}}

\nc{\query}{\ensuremath{\phi}}
\nc{\qatom}{\ensuremath{\alpha}}

\nc{\HG}{\ensuremath{\mathcal{H}}} %
\nc{\Nodes}{\ensuremath{V}} %
\nc{\Edges}{\ensuremath{E}} %

\nc{\HD}{\ensuremath{\textit{H}}} %
\nc{\Tree}{\ensuremath{T}}
\nc{\rootedTree}{\ensuremath{\hat{\Tree}}}

\nc{\treenode}{\ensuremath{t}} %
\nc{\treenodeparent}{\ensuremath{p}} %
\nc{\parentnode}{\treenodeparent} 
\nc{\treeroot}{\ensuremath{r}} %
\nc{\Bag}{\ensuremath{\chi}}
\nc{\Cover}{\ensuremath{\lambda}}

\nc{\FHD}{\ensuremath{\textit{F}}} %
\nc{\Weight}{\ensuremath{\gamma}} %

\nc{\Width}{\ensuremath{\textit{width}}}
\nc{\GHW}{\ensuremath{\textit{ghw}}} %
\nc{\FHW}{\ensuremath{\textit{fhw}}} %

\nc{\freetreenodes}{\ensuremath{U}} %

\nc{\prunedTree}{\ensuremath{\tilde{\Tree}}}
\nc{\prunedSchema}{\ensuremath{\tilde{\schema}}}
\nc{\prunedDB}{\ensuremath{\tilde{\DB}}}
\nc{\prunedDBold}{\ensuremath{\prunedDB_{\textit{old}}}}
\nc{\prunedDBnew}{\ensuremath{\prunedDB_{\textit{new}}}}
\nc{\prunedQuery}{\ensuremath{\tilde{\query}}}

\nc{\Start}{\ensuremath{\texttt{fetch-first}}}
\nc{\Next}{\ensuremath{\texttt{fetch-next}}}
\nc{\TestTuple}{\ensuremath{\texttt{test}}}
\nc{\PositionCursor}{\ensuremath{\texttt{position-cursor}}}

\nc{\Mapping}[1]{\ensuremath{\tilde{#1}}} %
\nc{\MappingR}{\ensuremath{\tilde{R}}}

\newcommand{\mapR}{\ensuremath{\mathbf{R}}}

\nc{\Algo}[1]{\ensuremath{\textsc{#1}}}

\nc{\True}{\ensuremath{\texttt{true}}}
\nc{\False}{\ensuremath{\texttt{false}}}

\newcommand{\widthw}{w}

\newcommand{\SubWidthSet}{\ensuremath{\mathsf S}} %
\newcommand{\TreeDecompSet}{\ensuremath{\mathsf T}} %
\newcommand{\fcTreeDecompSet}{\ensuremath{\mathsf f\mathsf c\mathsf T}} %
\nc{\SUBW}{\ensuremath{\textit{subw}}} %
\nc{\ADW}{\ensuremath{\textit{adw}}} %
\nc{\fcSUBW}{\ensuremath{\textit{fc-subw}}} %
\nc{\fcFHW}{\ensuremath{\textit{fc-fhw}}} %
\nc{\fcGHW}{\ensuremath{\textit{fc-ghw}}} %

\nc{\TD}{\ensuremath{\textit{TD}}} %
\nc{\FDecom}{\ensuremath{\textit{F}}} %

\newcommand{\projvar}[2]{#1\langle#2\rangle}

\newcommand{\paramell}{\ell}

\newcommand{\vsets}{\mathfrak{s}}

\newcounter{CommentCounterRed}
\newcounter{CommentCounterBlue}
\newcounter{CommentCounterGreen}
\newcounter{CommentCounterGray}
\newcounter{CommentCounterMagenta}

\begin{document}

\maketitle

\begin{abstract}
Marx (STOC~2010, J.~ACM 2013) introduced the notion of {submodular width} of a conjunctive query (CQ) and showed that for any class $\Phi$ of Boolean CQs of bounded submodular width, the model-checking problem for $\Phi$ on the class of all finite structures is fixed-parameter tractable (FPT). 
Note that for non-Boolean queries, the size of the query result may be far too large to be computed entirely within FPT time.
We investigate the free-connex variant of submodular width and
generalise Marx's result to non-Boolean queries as follows: For every class $\Phi$ of CQs of bounded free-connex submodular width, within FPT-preprocessing time we can build a data structure that allows to enumerate, without repetition and with constant delay, all tuples of the query result. 
Our proof builds upon Marx's splitting routine to decompose the query result into a union of results; but we have to tackle the additional technical difficulty to ensure that these can be enumerated efficiently.
\end{abstract}

\section{Introduction}\label{section:intro}

In the past decade, starting with Durand and
Grandjean~\cite{DurandGrandjean_BoundedDegree}, the fields of 
logic in computer science and database theory 
have seen a large number of
contributions that deal with the efficient enumeration of query
results. In this scenario, the objective is as follows: given a
finite relational structure (i.e., a database) and a logical formula (i.e., a
query), after a short preprocessing phase, the query results shall be
generated one by one, without repetition, with guarantees on the
maximum delay time between the output of two tuples.  In this vein,
the best that one can hope for is \emph{constant} delay (i.e., the
delay may depend on the size of the query but not on that of the
input structure) and \emph{linear} preprocessing time 
(i.e., time $f(\phi){\cdot} \bigOh(N)$ where $N$ is the size of a reasonable representation of the input structure, $\phi$ is the query, and $f(\phi)$ is a number only depending on the query but not on the input structure).  Constant delay
enumeration has also been adopted as a central concept in 
\emph{factorised databases} that gained recent
attention~\cite{DBLP:journals/tods/OlteanuZ15,DBLP:journals/sigmod/OlteanuS16}.

Quite a number of query evaluation problems are known to admit
constant delay algorithms preceded by linear or pseudo-linear time
preprocessing.  This is the case for all first-order queries, provided
that they are evaluated over
classes of structures of bounded
degree~\cite{DurandGrandjean_BoundedDegree,KazanaSegoufin_BoundedDegree,BKS-ICDT17,DBLP:conf/lics/KuskeS17},
low degree~\cite{DBLP:conf/pods/DurandSS14}, 
bounded expansion \cite{DBLP:conf/pods/KazanaS13},
locally bounded expansion~\cite{DBLP:conf/icdt/SegoufinV17},
and on
classes that are nowhere dense \cite{DBLP:conf/pods/SchweikardtSV18}.  
Also different data models have been investigated, including tree-like
data and document spanners 
\cite{%
DBLP:conf/csl/Bagan06,
DBLP:journals/tocl/KazanaS13,
DBLP:conf/icdt/AmarilliBMN19}.
Recently, also the \emph{dynamic} setting, where a fixed query has to be
evaluated repeatedly against a database that is constantly updated,
has received quite some attention
\cite{%
DBLP:conf/csl/LosemannM14,
BKS-ICDT17,
BKS_enumeration_PODS17,
DynamicYannakakis2017,
DBLP:conf/icdt/BerkholzKS18,
DBLP:conf/icdt/AmarilliBM18,
DBLP:conf/pods/NiewerthS18,
DBLP:conf/lics/Niewerth18,
DBLP:conf/pods/AmarilliBMN19}.

This paper deals with the classical, \emph{static} setting without
database updates.
We focus on evaluating conjunctive queries (CQs, i.e.,
primitive-positive formulas) on
arbitrary relational structures.\footnote{In this paper, structures
will always be finite and relational. }
In the following, \emph{FPT-preprocessing} (resp., \emph{FPL-preprocessing}) means preprocessing that takes time $f(\phi){\cdot}N^{\bigOh(1)}$ (resp.,
$f(\phi){\cdot}\bigOh(N)$), and \emph{constant delay} means delay $f(\phi)$, where $f$ is a computable function, $\phi$ is the query, and $N$ is the size of the input structure.

Bagan et al.~\cite{Bagan.2007} showed that
every free-connex acyclic CQ allows
constant delay enumeration after FPL-preprocessing. 
More refined results in this vein are due to
Bagan~\cite{Bagan_PhD} and Brault-Baron~\cite{BraultBaron_PhD}; see
\cite{DBLP:journals/sigmod/Segoufin15} for a
survey and \cite{SiglogNewsTutorial_BGS2020} for a tutorial.  
Bagan et al.\ \cite{Bagan.2007} complemented their result by a
conditional lower bound: 
assuming that Boolean matrix multiplication cannot be accomplished in
time $O(n^2)$, 
self-join-free acyclic CQs that are \emph{not} free-connex
\emph{cannot} be enumerated with constant
delay and FPL-preprocessing. 
This demonstrates that even if the evaluation of \emph{Boolean}
queries is easy (as known for all acyclic CQs
\cite{Yannakakis1981}), the enumeration of the results of
non-Boolean queries might be hard (here, for acyclic CQs that are not free-connex).

Bagan et al.\ \cite{Bagan.2007} also introduced the notion of \emph{free-connex} (fc)
treewidth (tw) of a CQ and showed that
for every class $\Phi$ of CQs of bounded fc-tw, within FPT-preprocessing time, one can build a data structure that allows constant delay enumeration of the query results.
This can be viewed as a generalisation, to the non-Boolean case, of the 
well-known result 
stating that
the model-checking problem for classes of Boolean CQs of bounded
treewidth
is 
FPT.
Note that for non-Boolean queries---even if they come from a class of
bounded fc-tw---the size of the query result may be
$N^{\Omega(\size{\phi})}$, i.e., far too large to be computed entirely within FPT-preprocessing time; and generalising the known tractability result for Boolean CQs to the non-Boolean case is far from trivial.

In a series of papers, the FPT-result for Boolean CQs has been strengthened to more and more general width-measures, namely to classes of queries of bounded generalised hypertree width (ghw) \cite{DBLP:journals/jcss/GottlobLS02}, bounded fractional hypertree width (fhw) \cite{DBLP:journals/talg/GroheM14}, and bounded submodular width (subw) \cite{Marx.2013}.
The result on bounded fhw has been generalised to the non-Boolean case
in the context of factorised databases
\cite{DBLP:journals/tods/OlteanuZ15}, which implies constant delay
enumeration after FPT-preprocessing for CQs of bounded free-connex
fractional hypertree width (fc-fhw). Related data structures that allow
constant delay enumeration after FPT-preprocessing %
for 
(quantifier-free) CQs of
bounded 
(fc-)fhw
have 
also
been provided in
\cite{DBLP:conf/pods/DeepK18, DBLP:conf/icdt/KaraO18}.

An analogous generalisation of the result on bounded submodular width, however, is still missing. The present paper's main result closes this gap: we show that 
on classes of CQs of bounded fc-subw, within FPT-preprocessing time one can build a data structure that allows constant delay enumeration of the query results.
And within the same FPT-preprocessing time, one can also construct 
a data structure that enables to test in
constant time whether an input tuple belongs to the query result.
Our proof uses Marx's splitting routine \cite{Marx.2013} to decompose the query result
of $\phi$ on $\A$ into
the union of results of several queries $\phi_i$ on several structures $\A_i$
but we have to tackle
the additional technical difficulty to ensure that the results of all the
$\phi_i$ on $\A_i$ can be enumerated efficiently. 
Once having achieved
this, we can conclude by using an elegant trick 
provided by Durand and Strozecki \cite{DBLP:conf/csl/DurandS11} for
enumerating, without repetition, the union of query results. 

As an immediate consequence of the lower bound provided by Marx
\cite{Marx.2013} in the context of Boolean CQs of unbounded submodular
width, one obtains that our main result is tight for certain classes
of CQs, namely, 
recursively enumerable
classes $\Phi$ of quantifier-free and
self-join-free CQs:
assuming the exponential time hypothesis (ETH), such a class
$\Phi$ allows constant delay enumeration after FPT-preprocessing if,
and only if, $\Phi$ has bounded fc-subw.

Let us mention a related recent result which, however, is incomparable to ours.
Abo~Khamis et~al.\ \cite{Khamis.2017}
designed an algorithm for evaluating a quantifier-free CQ $\phi$ of submodular width $w$ within time 
$\bigOh(N^w){\cdot} (\log N)^{f(\phi)} +  \bigOh(r{\cdot}\log N)$; and an analogous result is also achieved for non-quantifier-free CQs of fc-subw $w$ \cite{Khamis.2017}.
Here, $N$ is the size of the input structure, $r$ is the number of tuples in the query 
result, and $f(\phi)$ is at least exponential in number of variables of $\phi$.
In particular, the algorithm does not distinguish between a
preprocessing phase and an enumeration phase and does not provide a
guarantee on the delay.

\myparagraph{Outline}
The rest of the paper is structured as follows.
Section~\ref{section:Preliminaries} provides basic notations
concerning structures, queries, and constant delay enumeration.
Section~\ref{section:MainResult} recalls 
concepts of (free-connex) 
decompositions of queries, provides a precise statement of our main
result, and collects 
the necessary tools for obtaining this result.
Section~\ref{section:ProofOfMainResult} is devoted to the detailed proof of our main result.
We conclude in Section~\ref{section:FinalRemarks}.

\section{Preliminaries}\label{section:preliminaries}\label{section:Preliminaries}

In this section we fix notation and summarise basic definitions. 

\myparagraph{Basic notation}
We write $\NN$ and 
$\RRpos$ for the set of non-negative integers and reals, 
respectively,
and we let
$\NNpos\deff\NN\setminus\set{0}$ and $[n]\deff\set{1,\ldots,n}$ for
all $n\in\NNpos$.
By $\potenzmengeof{S}$ we denote the power set of a set $S$.
Whenever $G$ denotes a graph, we write $\Nodes(G)$ and $\Edges(G)$ for the set of
nodes and the set of edges, respectively, of $G$.
Whenever writing $a$
to denote a $k$-tuple (for some arity $k\in\NN$), we
write $a_i$ to denote the tuple's $i$-th component; i.e., $a=(a_1,\ldots,a_k)$.
For a $k$-tuple $a$ and indices $i_1,\ldots,i_\ell\in [k]$ we 
let $\proj_{i_1,\ldots,i_\ell}(a)\deff (a_{i_1},\ldots,a_{i_\ell})$.
For a set $S$ of $k$-tuples we let
$\proj_{i_1,\ldots,i_\ell}(S)\deff\setc{\proj_{i_1,\ldots,i_\ell}(a)}{a\in
  S}$.  

If $h$ and $g$ are mappings with domains $X$ and
$Y$, respectively, we
say that $h$ and $g$ are \emph{joinable} if
$h(z)=g(z)$ holds for all 
$z\in X\cap Y$. In case that $h$ and $g$ are joinable, we write 
$h \Join g$ to denote the mapping $f$ with domain
$X\cup Y$ where $f(x)=h(x)$ for all $x\in X$ and 
$f(y)=g(y)$ for all $y\in Y$.
If $A$ and $B$ are sets of mappings with domains $X$ and $Y$,
respectively, then $A\Join B
\deff\setc{h\Join g}{h\in A,\
  g\in B, \text{ and
    $h$ and $g$ are joinable} }$.

We use the following further notation where $A$ is a set of mappings with
domain $X$ and $h\in A$.
For a set $I\subseteq X$, the projection
$\proj_I(h)$ is the restriction $h_{|I}$ of $h$ to $I$; and
$\proj_I(A)\deff\setc{\proj_I(h)}{h\in A}$.
For objects $z,c$ where $z\not\in X$, we write $h\cup\set{(z,c)}$ for
the extension $h'$ of $h$ to domain $X\cup\set{z}$ with 
$h'(z)=c$ and $h'(x)=h(x)$ for all $x\in X$. 

\myparagraph{Signatures and structures}
A \emph{signature} is a finite set $\sigma$ of relation symbols, where
each $R\in\schema$ is equipped with a fixed \emph{arity} $\ar(R)\in\NNpos$.
A \emph{$\sigma$-structure} $\A$ consists of a finite set $A$
(called the \emph{universe} or \emph{domain} of $\A$)
and an $\ar(R)$-ary relation $R^{\A}\subseteq
A^{\ar(R)}$ for each $R\in\sigma$.
The \emph{size} $\size{\sigma}$ of a signature $\sigma$ is 
$|\sigma|+\sum_{R\in\sigma}\ar(R)$.
We write $n^\A$ to denote the cardinality $|A|$ of $\A$'s universe,
we write $m^\A$ to denote the number of tuples in $\A$'s largest
relation, 
and we write $N^\A$ or $\size{\A}$ to denote the size of a reasonable
encoding of $\A$. To be specific, let 
$N^\A=\size{\A}=\size{\sigma}+n^\A +\sum_{R\in\sigma} \size{R^\A}$,
where $\size{R^\A}=\ar(R){\cdot}|R^{\A}|$.
Whenever $\A$ is clear from the context, we will omit the superscript
$\cdot^{\A}$ and write $n,m,N$ instead of $n^\A,m^\A,N^\A$.
Consider signatures $\sigma$ and $\tau$ with $\sigma\subseteq \tau$.
The \emph{$\sigma$-reduct} of a $\tau$-structure $\B$ is the
$\sigma$-structure $\A$ with $A=B$ and $R^\A=R^\B$ for all
$R\in\sigma$.
A \emph{$\tau$-expansion} of a $\sigma$-structure $\A$ is a
$\tau$-structure $\B$ whose $\sigma$-reduct is $\A$.

\myparagraph{Conjunctive Queries}
We fix a countably infinite set $\Var$ of \emph{variables}.
We allow queries to use arbitrary relation symbols of arbitrary arities.
An \emph{atom} $\qatom$ is of the form 
$R(v_1, \ldots, v_r)$ with $r=\ar(R)$ and
$v_1,\ldots,v_r\in\Var$.
We write $\Vars(\qatom)$ to denote the set of variables occurring in $\qatom$.
A \emph{conjunctive query} (CQ, for short) is of the form 
\;$  \exists z_1 \cdots \exists z_\ell \,\big(
   \qatom_1 \uund \cdots \uund \qatom_d \,\big)$,
where $\ell\in\NN$, $d\in \NNpos$, $\qatom_j$ is an atom
for every $j\in [d]$, and
$z_1,\ldots,z_\ell$ are pairwise distinct elements in $\vars(\qatom_1)\cup\cdots\cup\vars(\qatom_d)$.
For such a CQ $\query$
we let $\atoms(\query)=\set{\qatom_1,\ldots,\qatom_d}$. 
We write
$\Vars(\query)$ and $\sigma(\query)$ for the set of variables and
the set of relation symbols occurring in
$\query$, respectively. 
The set of \emph{quantified} variables of $\query$ is
$\quant(\query)\deff \set{z_1,\ldots,z_\ell}$, and 
the set of \emph{free} variables is
$\free(\query)\deff \Vars(\query)\setminus \quant(\query)$.
We sometimes write $\phi(x_1,\ldots,x_k)$ to indicate that
$x_1,\ldots,x_k$ are the free variables of $\phi$.
The \emph{arity} of $\phi$ is the number $k\deff |\free(\phi)|$.
The query $\query$ 
is called \emph{quantifier-free} if $\quant(\query)=\emptyset$, 
it is called \emph{Boolean} if its arity is 0, and
it is called \emph{self-join-free} if no relation symbol occurs more than once in $\query$.

The semantics are defined as usual: \
A \emph{valuation} for $\query$ on a $\sigma(\phi)$-structure $\A$ is
a mapping 
$\valuation:\Vars(\query)\to A$.
A valuation $\valuation$ is a \emph{homomorphism} from $\query$ to a $\A$
if for every atom $R(v_1,\ldots, v_r)\in\atoms(\query)$  we have 
$\big(\valuation(v_1),\ldots,\valuation(v_r)\big)\in R^\A$. 
The \emph{query result} $\eval{\query}{\A}$ of 
a CQ $\phi$ 
on the $\sigma(\phi)$-structure $\A$ is
defined as the set 
$\{\,\proj_{\free(\phi)}(\valuation) \ : \ \text{$\valuation$ is a
  homomorphism from $\query$ to $\A$}\}$.
Often, we will identify the mappings $g\in\eval{\query}{\A}$ with
tuples $(g(x_1),\ldots,g(x_k))$, where $x_1,\ldots,x_k$ is a fixed
listing of the free variables of $\phi$.

The size $\size{\query}$ of a query $\query$ is 
the length of $\query$ when viewed as a word over the alphabet 
$\sigma(\phi)\cup\vars(\phi)\cup\set{\exists,\,\und\,,(\,,)\,}\cup\set{\,,}$.

\myparagraph{Model of computation} For the complexity analysis we
assume the RAM-model with a uniform cost measure. In particular,
storing and accessing elements from a structure's universe requires  $O(1)$ space and
time.
For an $r$-ary relation $R^{\A}$ we can construct in
time $O(\|R^{\A}\|)$ an index that allows to enumerate $R^{\A}$ with
$O(1)$ delay and to test for a given $r$-tuple $a$
whether $a\in R^{\A}$ in time $O(r)$.
Moreover, for every $\{i_1,\ldots,i_\ell\}\subseteq [r]$ we can build
 a data structure where we can enumerate for every $\ell$-tuple $b$
the selection $\setc{a\in
  R^{\A}}{\proj_{i_1,\ldots,i_\ell}(a)=b}$  with $O(1)$
delay. 
Such a data structure can be constructed in time $O(\|R^{\A}\|)$, for
instance by a linear scan over
$R^{\A}$ where we add every tuple $a\in R^{\A}$ to a list
$\mathcal L_{\proj_{i_1,\ldots,i_\ell}(a)}$. Using a constant
access data structure of linear size, the list $\mathcal L_{b}$
can be accessed in
time $O(\ell)$ when receiving an $\ell$-tuple $b$. 

\myparagraph{Constant delay enumeration and testing}
An \emph{enumeration algorithm} for query evaluation consists of two phases:
the preprocessing phase and the enumeration phase. In the
preprocessing phase the algorithm is allowed to do arbitrary
preprocessing on the query $\phi$ and the input structure $\A$. 
We denote the time required for this phase by
$\preprocessingtime$.
In the subsequent enumeration phase the algorithm enumerates, without repetition, all
tuples (or, mappings) in the query result $\eval{\phi}{\A}$, 
followed by the end-of-enumeration message
$\EOE$. The \emph{delay} $\delaytime$ is the maximum time that passes
between the start of the enumeration phase and the output of the first
tuple, between the output of two consecutive tuples, and between the
last tuple and \EOE. 

A \emph{testing algorithm} for query evaluation also starts with a
preprocessing phase of time $\preprocessingtime$ in which a data
structure is computed that allows to test for a given tuple (or,
mapping) $b$
whether it is contained in the query result $\eval{\phi}{\A}$. The \emph{testing time}
$\testingtime$ of the algorithm is an upper bound on the time that passes
between receiving $b$ and providing the answer.

One speaks of \emph{constant} delay  (testing time) if the delay (testing time) depends
on the query $\phi$, but not on the input structure $\A$.

We make use of the following result from Durand and Strozecki, which
allows to efficiently enumerate the union of query results, provided
that each query result in the union can be enumerated 
and tested
efficiently. Note that
this is not immediate, because the union might contain 
many duplicates that need to be avoided during enumeration.

\begin{theorem}[\cite{DBLP:conf/csl/DurandS11}]\label{thm:UnionTrick}
  Suppose that there is an enumeration algorithm $\mathbb A$
  that receives a query $\query$ and a structure $\strucA$ and
  enumerates $\eval{\query}{\strucA}$ with delay $\delaytime(\query)$ after
  $\preprocessingtime(\query,\strucA)$ preprocessing time. 
  Further suppose that there is a testing algorithm $\mathbb B$
  that receives a query $\query$ and a structure $\strucA$ and
  has
  $\preprocessingtime(\query,\strucA)$ preprocessing time and $\testingtime(\query)$
  testing time.
  Then there is an algorithm $\mathbb C$ that receives $\ell$ queries
  $\query_i$ and structures $\strucA_i$ and allows to
  enumerate $\bigcup_{i\in[\ell]} \eval{\query_i}{\strucA_i}$
    with $O(\sum_{i\in[\ell]} \delaytime(\query_i)+\sum_{i\in[\ell]} \testingtime(\query_i))$ delay
  after $O(\sum_{i\in[\ell]} \preprocessingtime(\query_i,\strucA_i))$ preprocessing time.  
\end{theorem}

\begin{proof}[Proof (sketch)]
  The induction start $\ell=1$ is trivial.
  For the induction step  $\ell \to \ell+1$ start an enumeration of
  $\bigcup_{i\in[\ell]} \eval{\query_i}{\strucA_i}$ and test for every
  tuple whether it is contained in
  $\eval{\query_{\ell+1}}{\strucA_{\ell+1}}$. If the answer is no,
  then output the tuple. Otherwise discard the tuple and instead output the next tuple in an
  enumeration of
  $\eval{\query_{\ell+1}}{\strucA_{\ell+1}}$. Subsequently enumerate
  the remaining tuples from $\eval{\query_{\ell+1}}{\strucA_{\ell+1}}$.
\end{proof}

\section{Main Result}\label{section:MainResult}

At the end of this section, we provide a precise statement of our main
result.
Before we can do so, we have to recall the 
concept of free-connex 
decompositions of queries and the notion of submodular width.
It will be convenient for us to 
use
the following notation.

\begin{definition}
Let $\phi=\exists z_1 \cdots \exists z_\ell \,\big(
   \qatom_1 \uund \cdots \uund \qatom_d \,\big)$ be a CQ and $S\subseteq\vars(\phi)$.
We write
\,$\query\langle S\rangle$\,
 for the CQ that is equivalent to
       the expression
       \begin{equation}
         \label{eq:2}
   \big(\; \exists y_1 \cdots \exists y_r\  \qatom_1\;\big) \uund
   \cdots \uund \ \big(\; \exists
   y_1 \cdots \exists y_r\  \qatom_d \;\big),
       \end{equation}
where $\set{y_1,\ldots,y_r}=\vars(\phi)\setminus S$. 
\end{definition}

Note that $\query\langle S\rangle$ is
obtained from $\query$ by discarding existential quantification
and projecting every 
atom
to $S$, hence $\free(\query\langle S\rangle)=S$. However,
$\eval{\query\langle S\rangle}{\strucA}$ shall not be confused with the projection of
$\eval{\query}{\strucA}$ to $S$. In fact, it might be that
$\eval{\query}{\strucA}$ is empty, but $\eval{\query\langle
  S\rangle}{\strucA}$ is not, as the following example illustrates: 
\begin{align}
\query & \ = \ 
E(x,y)\land E(y,z)\land E(x,z) \text{ \quad and } \\
\query\langle \{x,z\}\rangle &\ \equiv \ \exists y
E(x,y)\land \exists y E(y,z)\land \exists
y E(x,z) \\
& \ \equiv\ 
E(x,z)\;.
\end{align}

\subsection{Constant delay enumeration using tree decompositions}

We use the same notation as \cite{GottlobEtAl_QandA} for decompositions of queries: 
A \emph{tree decomposition} (TD, for short) of a CQ $\query$ 
is a tuple $\TD=(\Tree,\Bag)$, for which the following two conditions are
satisfied:
\begin{enumerate}
\item 
 $\Tree=(\Nodes(\Tree),\Edges(\Tree))$ is a finite undirected tree.
\item
 $\Bag$ is a mapping that associates with every node $\treenode\in\Nodes(\Tree)$ a 
 set $\Bag(\treenode)\subseteq \Vars(\query)$ such that
 \begin{enumerate}
  \item
    for each atom $\qatom\in\Atoms(\query)$ there exists $\treenode\in\Nodes(\Tree)$ such that
    $\Vars(\qatom)\subseteq\Bag(\treenode)$, and
  \item\label{item:pathcondition:treedecomp}
    for each variable $v\in\Vars(\query)$ the set 
    $\Bag^{-1}(v)\deff\setc{\treenode\in\Nodes(\Tree)}{v\in\Bag(\treenode)}$ induces a 
    connected subtree of $\Tree$ 
(this condition is called \emph{path condition}). 
 \end{enumerate}
\end{enumerate}

To use a tree decomposition $\TD=(\Tree,\Bag)$ of $\phi$ for query
evaluation one considers, for
each 
$\treenode\in\Nodes(\Tree)$ the query 
$\query\langle S\rangle$ for $S\deff \Bag(\treenode)$, evaluates this
query on the input structure $\A$, and then combines these results for
all $\treenode\in\Nodes(\Tree)$ along a bottom-up traversal of
$\Tree$.
If the query is Boolean, this yields the result of $\phi$ on $\A$; if
it is non-Boolean, $\eval{\phi}{\A}$ can be computed by performing
additional traversals of $\Tree$.
This approach is efficient if the result sets
$\eval{\phi\langle\Bag(\treenode)\rangle}{\A}$
are small 
and can be computed efficiently
(later on, we will sometimes refer to 
the sets $\eval{\phi\langle\Bag(\treenode)\rangle}{\A}$
as \emph{projections on bags}).

The simplest queries where this is the case are acyclic queries
\cite{DBLP:journals/jacm/BeeriFMY83,DBLP:journals/siamcomp/BernsteinG81}. 
A number of equivalent
characterisations of the acyclic CQs have been provided in the
literature (cf.\ 
\cite{%
AHV-Book,
DBLP:journals/jcss/GottlobLS02,
DynamicYannakakis2017,
DBLP:journals/csur/Brault-Baron16%
}); among them a
characterisation by Gottlob et al.\
\cite{DBLP:journals/jcss/GottlobLS02} stating that a CQ is acyclic if
and only if it has
a tree-decomposition where every bag is covered by an atom, i.e., for every bag $\Bag(t)$ there is some atom
$\qatom$ in $\query$ with $\Bag(t)\subseteq\vars(\qatom)$.
The approach described above leads to a linear time algorithm for
evaluating an acyclic CQ $\phi$ that is Boolean, and if $\phi$ is
non-Boolean, 
$\eval{\query}{\strucA}$ is computed in time 
linear in $\size{\strucA}+|\,\eval{\query}{\strucA}|$.
This method is known as \emph{Yannakakis' algorithm}.
But this algorithm does not distinguish between a preprocessing phase
and an enumeration phase and does not guarantee constant delay enumeration.
In fact, Bagan et al.\ identified the following additional property
that is needed to ensure constant delay enumeration.
 
\begin{definition}[\cite{Bagan.2007}]
  A tree decomposition $\TD=(\Tree,\Bag)$  of
  a CQ $\query$ is
  \emph{free-connex} if there is a subset $\freetreenodes \subseteq
  \Nodes(\Tree)$ 
  that induces a connected subtree of $\Tree$ and that satisfies the condition
  $\free(\query) =
  \bigcup_{\treenode\in\freetreenodes} \Bag(\treenode)$.
\end{definition}

Bagan et al.\ \cite{Bagan.2007} identified the \emph{free-connex}
acyclic CQs, i.e., the CQs $\query$ that have a free-connex
tree decomposition where every bag is covered by an atom, as the fragment of the acyclic
CQs whose results can be enumerated with constant delay after
FPL-preprocessing:

\begin{theorem}[Bagan et al.~\cite{Bagan.2007}]\label{thm:BaganEtAl}\label{thm:BaganEtAlRefined}
There is a computable function $f$ and an algorithm 
which receives a free-connex acyclic CQ
$\query$ and a 
$\sigma(\query)$-structure $\strucA$ 
and computes within
$\preprocessingtime=
f(\query)O(\size{\A})$
preprocessing time 
and space
a data structure that allows to
\begin{enumerate}[(i)]
\item
 enumerate $\eval{\query}{\strucA}$ with $f(\query)$ delay and
\item
 test for a given tuple (or, mapping) $b$ if $b\in\eval{\query}{\strucA}$ within
 $f(\query)$ testing time.
\end{enumerate}
\end{theorem}

The approach of using free-connex tree decompositions for constant
delay enumeration can be extended from acyclic CQs to arbitrary
CQs. To do this, we have to compute for every bag $\Bag(t)$ in the
tree decomposition the projection
$\eval{\query\langle\Bag(t)\rangle}{\strucA}$. This reduces
the task to
the acyclic case, where the free-connex acyclic query contains one atom $\qatom$ with
$\vars(\qatom)=\Bag(t)$ for every bag $\Bag(t)$ and the corresponding
relation is defined by $\eval{\query\langle\Bag(t)\rangle}{\strucA}$.
Because the 
runtime 
in this approach is dominated by computing
$\eval{\query\langle\Bag(t)\rangle}{\strucA}$, it is  only
feasible if the projections are efficiently computable for every bag.
If the decomposition has bounded treewidth or bounded fractional
hypertree width, then it is possible to compute
$\eval{\query\langle\Bag(t)\rangle}{\strucA}$ for every bag in time
$f(\phi){\cdot}\size{\A}^{\bigOh(1)}$ 
\cite{DBLP:journals/talg/GroheM14},
which in turn implies that the result can be enumerated after
FPT-preprocessing time for CQs of bounded fc-tw 
\cite{Bagan.2007} and for CQs of bounded fc-fhw \cite{DBLP:journals/tods/OlteanuZ15}.

\subsection{Submodular width and statement of the main result}\label{subsection:subw}

Before providing the precise definition of the submodular width of a
query, let us first consider an example.
The central idea behind algorithms that rely on submodular width
\cite{Marx.2013,Khamis.2017,Scarcello.2018} is to
split the input structure into several parts and use for every
part a different tree decomposition of $\query$.
This will give a significant improvement over the fractional
hypertree width, which uses only one tree decomposition of $\query$.
A typical example to illustrate this idea is the following $4$-cycle
query (see also \cite{Khamis.2017,Scarcello.2018}):
$
  \query_4 :=
   E_{12}(x_1,x_2) \uund
   E_{23}(x_2,x_3) \uund
   E_{34}(x_3,x_4) \uund
   E_{41}(x_4,x_1).
$

There are essentially two non-trivial tree decompositions
$\TD'=(\Tree,\Bag')$, $\TD''=(\Tree,\Bag'')$  of $\query_4$,
which are both defined over the two-vertex tree
$\Tree=(\{\treenode_1,\treenode_2\},\{(\treenode_1,\treenode_2)\})$ by
$\Bag'(\treenode_1) = \{x_1,x_2,x_3\}$, $\Bag'(\treenode_2) =
\{x_1,x_3,x_4\}$ and 
 $\Bag''(\treenode_1) = \{x_2,x_3,x_4\}$, $\Bag''(\treenode_2) = \{x_1,x_2,x_4\}$.
Both tree decompositions lead to an optimal fractional hypertree
decomposition of width $\FHW(\query_4)=2$. Indeed, for the worst-case
instance $\strucA$ with 
\begin{align*}
  E_{12}^{\strucA} &= E_{34}^{\strucA} := [\paramell] \times
  \{a\} \  \cup\ \{b\}\times [\paramell] \quad
  &
  E_{23}^{\strucA} &= E_{41}^{\strucA} := [\paramell] \times \{b\}\  \cup\ \{a\}\times [\paramell] 
\end{align*}
we have $\|\strucA\| = O(\paramell)$ while the projections on the bags have size $\Omega(\paramell^2)$ in
both decompositions:\footnote{recall from
  Section~\ref{section:preliminaries} our convention to identify 
  mappings in query results with tuples; the free variables are
  listed canonically here, by increasing indices} 
\begin{align*}
  \eval{\projvar{\query_4}{\Bag'(\treenode_1)}}{\strucA} =
  \eval{\projvar{\query_4}{\Bag'(\treenode_2)}}{\strucA} &=
                                                 [\paramell]\times\{a\}\times[\paramell]
                                                \ \cup\ \{b\}\times
                                                [\ell]\times \{b\},\\
  \eval{\projvar{\query_4}{\Bag''(\treenode_1)}}{\strucA} =
  \eval{\projvar{\query_4}{\Bag''(\treenode_2)}}{\strucA} &= [\paramell]\times\{b\}\times[\paramell]
                                                \ \cup\  \{a\}\times
                                                [\ell]\times \{a\}.  
\end{align*}
However, we can split $\strucA$ into $\strucA'$ and $\strucA''$ such that
$\eval{\query_4}{\strucA}$ is the disjoint union of $\eval{\query_4}{\strucA'}$ and
$\eval{\query_4}{\strucA''}$ and the bag-sizes in the respective decompositions
are small:
\begin{align*}
  E_{12}^{\strucA'} &= E_{34}^{\strucA'} :=  \{b\}\times [\paramell]
  &
  E_{23}^{\strucA'} &= E_{41}^{\strucA'} := [\paramell] \times \{b\}\\
  E_{12}^{\strucA''} &= E_{34}^{\strucA''} := [\paramell] \times \{a\}
  &
  E_{23}^{\strucA''} &= E_{41}^{\strucA''} := \{a\}\times [\paramell]\\ 
  \eval{\projvar{\query_4}{\Bag'(\treenode_1)}}{\strucA'} &=
  \eval{\projvar{\query_4}{\Bag'(\treenode_2)}}{\strucA'} =
                                                \{b\}\times
                                                [\ell]\times \{b\},\\
  \eval{\projvar{\query_4}{\Bag''(\treenode_1)}}{\strucA''} &=
  \eval{\projvar{\query_4}{\Bag''(\treenode_2)}}{\strucA''} =  \{a\}\times
                                                [\ell]\times \{a\}.  
\end{align*}
Thus, we can efficiently evaluate $\query_4$ on $\strucA'$ using $\TD'$
and $\query_4$ on $\strucA''$ using $\TD''$ (in time $O(\ell)$ in this example) and then combine both results
to obtain $\query_4(\strucA)$.
Using the strategy of Alon~et~al.~\cite{Alon.1997}, it is possible to
split \emph{every} database $\strucA$ for this particular 4-cycle query
$\query_4$ into two instances $\strucA'$ and $\strucA''$ such that
the bag sizes in $\TD'$ on $\strucA'$ as well as in $\TD''$ on $\strucA''$ are bounded by
$\|\strucA\|^{3/2}$ and can be computed in time $O(\|\strucA\|^{3/2})$ (see
\cite{Khamis.2017,Scarcello.2018} for a detailed account on this strategy). As
both decompositions are free-connex, this also leads to a constant
delay enumeration algorithm for $\query_4$ with $O(\|\strucA\|^{3/2})$
time preprocessing, which improves the $O(\|\strucA\|^{2})$ preprocessing time
that follows from using one decomposition.

In general, whether such a data-dependent decomposition is possible is
determined by the submodular width $\SUBW(\query)$ of the query.
The notion of submodular width
was introduced in \cite{Marx.2013}. To present its definition, we need the following terminology.
A function $g\colon 2^{\Vars(\query)} \to \RRpos$ 
is
\begin{itemize}
\item \emph{monotone} if $g(U) \leq g(V)$ for all $U\subseteq V
  \subseteq \Vars(\query)$.
\item \emph{edge-dominated} if $g(\vars(\qatom))\leq 1$ 
 for every atom $\qatom\in\atoms(\query)$.
\item \emph{submodular}, if $g(U) + g(V) \geq g(U\cap V) + g(U\cup V)$ for
every $U,V\subseteq \Vars(\query)$.
\end{itemize}
We denote by $\SubWidthSet(\query)$ the set of all monotone, edge-dominated,
submodular functions  $g\colon 2^{\Vars(\query)} \to \RRpos$ that
satisfy $g(\emptyset) = 0$, and by $\TreeDecompSet(\query)$ the set of all
tree decompositions of $\query$.
The \emph{submodular width} of a conjunctive query
  $\query$ is
  \newcommand{\textstylehere}{}
  \begin{equation}
    \label{eq:defsubw}
    \SUBW(\query) \ \ := \ \ \textstylehere\sup_{g\in
      \SubWidthSet(\query)}\;\textstylehere\min_{(\Tree,\Bag)\in\TreeDecompSet(\query)}\;\textstylehere\max_{t\in\Nodes(\Tree)}\
    g(\Bag(\treenode)).
  \end{equation}
In particular, if the submodular width of $\query$ is bounded
by $\widthw$, then for every submodular function $g$ there is a tree
decomposition in which every bag $B$ satisfies $g(B)\leq \widthw$.
It is known that $\SUBW(\query)\leq
\FHW(\query)$ for all queries $\query$
\cite[Proposition~3.7]{Marx.2013}.
Moreover, there is a
constant $c$ and a
family of queries $\query$
such that $\SUBW(\query)\leq c$ is bounded and
$\FHW(\query)=\Omega(\sqrt{\log \|\query\|})$ is unbounded \cite{Marx.2011,Marx.2013}.
The main result in \cite{Marx.2013} is that the submodular width
characterises the tractability of Boolean CQs in the following sense.

\begin{theorem}[\cite{Marx.2013}]\label{thm:marx}\mbox{\ }  
  \begin{enumerate}[(1)]
  \item\label{item:algo:thm:marx}
 There is a computable function $f$ and an algorithm that receives 
 a Boolean %
 CQ $\phi$, $\SUBW(\query)$, and a 
  $\sigma(\query)$-structure $\A$
  and evaluates $\query$ on $\strucA$ in time
    $f(\query)\size{\A}^{O(\SUBW(\query))}$.
  \item 
   Let $\Phi$ be a recursively enumerable class of Boolean,
   self-join-free CQs of unbounded submodular width.
    Assuming the exponential time hypothesis (ETH) there is no
    algorithm which, upon input of a query $\query\in\Phi$ and a
    structure $\A$, evaluates $\query$ on $\A$ in time
    $\size{\A}^{o(\SUBW(\query)^{1/4})}$.
  \end{enumerate}
\end{theorem}

The \emph{free-connex submodular width} of a conjunctive query
  $\query$ is defined in a similar way as submodular width, but this time ranges over the
  set $\fcTreeDecompSet(\query)$ of all free-connex tree
  decompositions of $\query$ 
(it is easy to see that
 we can assume that $\fcTreeDecompSet(\query)$ is finite).
  \begin{equation}
    \label{eq:deffcsubw}
    \fcSUBW(\query) := \textstylehere\sup_{g\in
      \SubWidthSet(\query)}\;\textstylehere\min_{(\Tree,\Bag)\in\fcTreeDecompSet(\query)}\;\textstylehere\max_{t\in\Nodes(\Tree)}\
    g(\Bag(\treenode)).
  \end{equation}

Note that if $\query$ is either quantifier-free or Boolean, we have
$\fcSUBW(\query)=\SUBW(\query)$. In general, this is not always the
case. Consider for example the following quantified version $\query'_4 := \exists x_1\exists
x_3\, \query_4$ of the quantifier-free 4-cycle query $\query_4$. Here we have $\SUBW(\query'_4) = \frac32$,
but $\fcSUBW(\query'_4) = 2$: one can verify $\fcSUBW(\query'_4) \geq
2$ by noting that every free-connex tree decomposition contains a bag
$\{x_1,x_2,x_3,x_4\}$ and taking the submodular function $g(U) := \frac12|U|$.
Now we are ready to state the main theorem of this paper.

\begin{theorem}\label{thm:fcsubw-with-testing}
For every $\delta > 0$ and $\widthw\geq 1$ there is
a computable function $f$ and
 an algorithm which
receives a CQ
$\query$ with $\fcSUBW(\query)\leq \widthw$ and a
$\sigma(\query)$-structure $\strucA$ 
and computes within
$\preprocessingtime=
f(\query)\size{\A}^{(2+\delta)\widthw}$
preprocessing time 
and space
$f(\query)\size{\A}^{(1+\delta)\widthw}$ a data structure that allows to
\begin{enumerate}[(i)]
\item
 enumerate $\eval{\query}{\strucA}$ with $f(\query)$ delay and
\item
 test for a given tuple (or, mapping) $b$ if $b\in\eval{\query}{\strucA}$ within
 $f(\query)$ testing time.
\end{enumerate}
\end{theorem}

\newcommand{\queryclass}{\ensuremath{\Phi}}

The following corollary is an immediate consequence of
Theorem~\ref{thm:fcsubw-with-testing} and Theorem~\ref{thm:marx}.
A class $\Phi$ of CQs is said to be of \emph{bounded free-connex submodular
width} if there exists a number $w$ such that $\fcSUBW(\phi)\leq w$ for
all $\phi\in\Phi$.
And by an \emph{algorithm for $\Phi$ that enumerates with constant delay
after FPT-preprocessing} we mean an algorithm that receives a
query $\phi\in \Phi$ and a $\sigma(\phi)$-structure $\A$ and
spends $f(\phi)\size{\A}^{O(1)}$ preprocessing time and then
enumerates $\eval{\phi}{\A}$ with delay $f(\phi)$, for a computable
function $f$.

\begin{corollary} \mbox{\ }
\begin{enumerate}[(1)]
\item
For every class $\Phi$ of CQs of bounded free-connex submodular width,
there is an algorithm for $\Phi$ that enumerates with constant delay
after FPT-preprocessing.
\item
Let $\queryclass$ be a 
recursively enumerable class of quantifier-free self-join-free
CQs and assume that the exponential time
hypothesis (ETH) holds. \\
  Then there is an
  algorithm for $\Phi$ that enumerates with constant delay after
  FPT-preprocessing if, and only if, $\queryclass$
  has bounded free-connex submodular width. 
\end{enumerate}
\end{corollary}

\section{Proof of the Main Result}\label{section:ProofOfMainResult}

To prove Theorem~\ref{thm:fcsubw-with-testing}, we  make use of Marx's
splitting routine for 
queries of bounded
submodular width.
In the following, we will adapt the main definitions and concepts from \cite{Marx.2013}
to our notions. While doing this, we
provide the following additional technical contributions: First, we give a
detailed time and space analysis of the algorithm and improve the
runtime of the consistency algorithm \cite[Lemma~4.5]{Marx.2013} from quadratic
to linear (see Lemma~\ref{lem:consistent}).
Second, we fix an oversight in \cite[Lemma~4.12]{Marx.2013} by establishing
strong $M$-consistency (unfortunately, this fix incurs a blow-up in
running time).
Afterwards we prove our main theorem, where the non-Boolean setting
requires us to relax Marx's partition into \emph{refinements}
(Lemma~\ref{lem:decompdisjoint}) so that the subinstances are no
longer disjoint.

Let $\query$ be a quantifier-free CQ 
with $\vars(\query)=\{x_1,\ldots, x_k\}$, and let 
$\sigma\deff \sigma(\phi)$. %
For every
$S = \{x_{i_1},\ldots,x_{i_{\ell}}\}\subseteq \vars(\query)$ where
$i_1<\cdots < i_{\ell}$ we set $x_S := (x_{i_1},\ldots,x_{i_{\ell}})$
and let $R_S\notin \sigma$ be 
a fresh
$\ell$-ary relation symbol.
For every
collection $\vsets \subseteq 2^{\vars(\query)}$ we let
\begin{align}
  \label{eq:7}
\sigma_{\vsets} &:= \sigma \cup \setc{R_S}{S\in\vsets} \quad\text{and}\\
\query_{\vsets} &:= \query
\,\land\,\textstyle\bigwedge_{S\in\vsets}R_S(x_S).
\end{align}
A \emph{refinement} of $\query$ and a $\sigma$-structure $\strucA$ is a pair $(\vsets,
\mathcal B)$, where $\vsets\subseteq 2^{\vars(\query)}$ is closed
under taking subsets and $\mathcal B$ is a $\sigma_{\vsets}$-expansion of $\strucA$.
Note that if $(\vsets,
\mathcal B)$ is a refinement of $\query$ and $\strucA$, then  $\eval{\query_{\vsets}}{\strucB} \subseteq
\eval{\query}{\strucA}$.
In the following we will construct refinements that do not change the
result relation, i.\,e., $\eval{\query_{\vsets}}{\strucB} =
\eval{\query}{\strucA}$.
Subsequently, we will split refinements in order to partition the
query result.

The following definition collects useful properties of
refinements. 
Recall from Section~\ref{section:preliminaries} that for a CQ $\psi$
and a structure $\B$, the query result $\eval{\psi}{\B}$ actually is a
set of mappings from $\free(\psi)$ to $B$.
For notational convenience we
define
$\mapR^\strucB_{S} := \eval{R_S(x_S)}{\strucB}$ and use the set
$\mapR^\strucB_{S}$ of 
mappings instead of the relation
$R^\strucB_{S}$.
In particular, by addressing/inserting/deleting a mapping $h\colon S\to B$ from $\mapR^\strucB_{S}$
we mean addressing/inserting/deleting the tuple
$(h(x_{i_1}),\ldots,h(x_{i_{\ell}}))$ from $R^\strucB_{S}$, where $(x_{i_1},\ldots,x_{i_{\ell}})=x_S$.

\begin{definition}
Let $\query$ be a quantifier-free $\sigma$-CQ, $\strucA$ a
$\sigma$-structure, $(\vsets,
\mathcal B)$ a refinement of $\query$ and $\strucA$, and $M$ an integer.
\begin{enumerate}
\item The refinement $(\vsets, \mathcal B)$ is
\emph{consistent} if
  \begin{align}
  \mapR_S^\strucB &= \eval{\query_{\vsets}\langle S\rangle}{\strucB} \text{ for all $S
\in \vsets$ and} \label{eq:cons1}\\
  \mapR^{\mathcal B}_S &= %
                     \pi_{S}\bigl(\mapR^{\strucB}_T\bigr)
                     \text{
   for all $S, T \in \vsets$ with $S\subset T$.}\label{eq:cons2}
  \end{align}
\item The refinement $(\vsets, \mathcal B)$ is
\emph{$M$-consistent} if it is consistent and 
  \begin{align}
    S\in \vsets 
\quad \Longleftrightarrow\quad \text{for all $T\subseteq S$: $|\,\eval{\query_{\vsets}\langle T\rangle}{\strucB}|\leq M$.}\label{eq:Mcons}
  \end{align}
  \item The refinement $(\vsets, \mathcal B)$ is \emph{strongly $M$-consistent} if it is $M$-consistent and
    \begin{align}
    S \in \vsets,\; T \in \vsets,\; (S\cup T) \notin \vsets\quad \Longrightarrow\quad
  \text{$|\,\eval{\query_{\vsets}\langle  S\cup T\rangle}{\strucB}| > M$.}\label{eq:stronglyMcons}
    \end{align}
\end{enumerate}

\end{definition}

\begin{lemma}\label{lem:consistent}
  There is an algorithm that receives a refinement $\mathcal R =
  (\vsets, \strucB)$ of $\query$ and $\strucA$ and computes in time $O(|\vsets|\cdot \|\strucB\|)$ a consistent refinement $(\vsets, \strucB')$ with $R^{\strucB'}_S\subseteq R^\strucB_S$ for all $S\in\vsets$
  and $\eval{\query_{\vsets}}{\strucB'} = \eval{\query_{\vsets}}{\strucB}$.
\end{lemma}

 \begin{proof}%
  We start by letting $\B'\deff \B$ and then proceed by iteratively
  modifying $\B'$.
   We first establish the first consistency requirement
   \eqref{eq:cons1} by removing  from every $\mapR^{\strucB'}_S$ all mappings $h$
   such that $h \notin \eval{\query_{\vsets}\langle S\rangle}{\strucB'}$. 
   To ensure the second consistency requirement  \eqref{eq:cons2},
   the algorithm iteratively deletes mappings in
   $\mapR^{\strucB'}_S$ that do not extend to larger mappings in $\mapR^{\strucB'}_T$ (for all
   $S\subset T \in\vsets$). Note that removing a mapping from $\mapR^{\strucB'}_T$ might shrink the set  
  $\eval{\query_{\vsets}\langle S'\rangle}{\strucB'}$ for sets
  $S'\in\vsets$ that have a nonempty intersection with $S$.
  In this case, we also have to delete affected mappings from
  $\mapR_{S'}^{\B'}$ in order to ensure that $\mapR_{S'}^{\B'}=\eval{\query_{\vsets}\langle S'\rangle}{\strucB'}$.   
  These steps will be iterated until the refinement is consistent. It
  is clear that the refinement does not exclude tuples from the query
  result, i.\,e., the final structure $\strucB'$ satisfies $\eval{\query_{\vsets}}{\strucB'} = \eval{\query_{\vsets}}{\strucB}$.
  To see that this can be achieved in time linear in
  $|\vsets|\cdot \sum_{S\in\vsets}|\mapR^\strucB_S|$, we formulate the
  problem as a set of
  Horn-clauses. The consistent refinement can then be computed by
  applying any linear-time unit propagation algorithm (cf., e.g.,
  \cite{DBLP:journals/jlp/DowlingG84}).
  For every $S\in\vsets$ and every mapping $h\in \mapR^\strucB_S$ we introduce a
  Boolean variable $d^h_S$ which expresses that, in order to achieve
  consistency, $h$ has to be deleted from $\mapR^\strucB_S$. The
  Horn-formula contains for every $S,T \in \vsets$ with $S \subset
  T$ the clauses
  \begin{align}
    d^g_S &\gets \textstyle\bigwedge \setc{d^h_T}{h\in \mapR^\strucB_T,\; \pi_S(h)
            = g}&&\text{for all $g\in \mapR^\strucB_S$, and}\\
    d^h_T &\gets d^g_S &&\text{for all $h\in \mapR^\strucB_T$, $g\in \mapR^\strucB_S$,
                          $\pi_S(h) = g$.}  
  \end{align}
  The first type of clauses ensures that when a mapping $g$ with
  domain $S$ does not
  extend to a tuple $h$ with domain $T\supset S$, then it will be
  excluded from 
$\mapR^{\strucB'}_S$. The
  second type of clauses ensures that for all $T\in\vsets$ we have
  $\mapR^{\strucB'}_T = \eval{\query_{\vsets}\langle T\rangle}{\strucB'}$.
  Note that the size of
  the resulting 
Horn-formula is bounded by $O\bigl(|\vsets|\cdot
  \sum_{S\in\vsets}|\mapR^\strucB_S|\bigr)$.
  Now we apply a linear time unit propagation algorithm to find a
  solution of minimum weight.
  If the formula is unsatisfiable, we know that 
  $\eval{\query_\vsets}{\strucB}=\emptyset$ and can safely set $\mapR^{\strucB'}_S=\emptyset$
  for all $S\in\vsets$. Otherwise, we obtain a minimal satisfying assignment $\beta$
  that sets a variable $d^h_S$ to true if, and only if, $h$ has to be
  deleted from $\mapR^\strucB_S$. 
  Thus we set  
$\mapR^{\strucB'}_S := \mapR^\strucB_S \setminus \setc{h}{\beta(d^h_S)=1}$. By
  minimality we have $\eval{\query_{\vsets}}{\strucB'} = \eval{\query_{\vsets}}{\strucB}$.
 \end{proof}

 \begin{lemma}\label{lem:strongMcons}
 Let $\query$ be a quantifier-free CQ, let $\strucA$ be
 a $\sigma(\query)$-structure where the largest relation contains $m$
 tuples, and let $M\geq m$.
   There is an algorithm that computes in time
   $O(2^{|\vars(\query)|}\cdot M^2)$ and
   space $O(2^{|\vars(\query)|}\cdot M)$ a strongly 
   $M$-consistent refinement $(\vsets,\strucB)$ that satisfies $\eval{\query}{\strucA} = \eval{\query_{\vsets}}{\strucB}$.
 \end{lemma}

    \begin{figure}
     \centering
\begin{algorithmic}[1]
 \State INPUT: %
quantifier-free CQ $\query(x_1,\ldots,x_k)$, $\sigma(\phi)$-structure $\strucA$
 \State $\vsets \gets \emptyset$\;;
 $\B \gets \A$
 \Repeat
   \For{$\ell = 1, \cdots, k$} \Comment{Step 1: Ensure condition \eqref{eq:Mcons}.}
   \For{$S=\{x_{i_1},\ldots,x_{i_\ell}\} \subseteq \vars(\query)$}
     \If{$S\notin \vsets$ and $S\setminus\{x\} \in \vsets$ for all $x\in S$}
     \State $\mapR^\strucB_S\gets\emptyset$
     \State Choose 
            $x\in S$ arbitrary
     \For{$h\in \mapR^\strucB_{S\setminus\{x\}}$ and 
            $c\in A$}
      \If{$h\cup \{(x,c)\}\in \eval{\query_{\vsets}\langle
   S\rangle}{\strucB}$}
      $\mapR^\strucB_S\gets \mapR^\strucB_S \cup \{h\cup \{(x,c)\}\}$
      \EndIf
     \EndFor
     \If{$|\mapR^\strucB_S| \leq M$}
     $\vsets \gets \vsets \cup \{S\}$
     \EndIf
     \EndIf
   \EndFor
   \EndFor
   \State
   \For{$S,T\in\vsets$ such that $S\cup T\notin\vsets$}
   \Comment{Step 2: Ensure condition \eqref{eq:stronglyMcons}.}
     \For{$g\in \mapR^\strucB_S$ and $h\in \mapR^\strucB_T$}
       \If{$g\Join h \in \eval{\query_{\vsets}\langle
   S\cup T\rangle}{\strucB}$}
        $\mapR^\strucB_{S\cup T} \gets \mapR^\strucB_{S\cup T} \cup
        \{g\Join h\}$
       \EndIf
        \If{$|\mapR^\strucB_{S\cup T}| > M$} \textbf{break} \EndIf
     \EndFor
        \If{$|\mapR^\strucB_{S\cup T}| \leq M$} 
        $\vsets \gets \vsets \cup \{S\cup T\}$
        \EndIf
   \EndFor
   \State
   \State $(\vsets, \strucB) \gets
   \Algo{Consistent}(\vsets, \strucB)$\Comment{Step 3: Apply Lemma~\ref{lem:consistent} to ensure
     \eqref{eq:cons1}, \eqref{eq:cons2}.}
 \Until{$\vsets$ remains unchanged}
   \State
 \Return{$(\vsets, \strucB)$}
\end{algorithmic}
     \caption{Computing a strongly $M$-consistent refinement}
     \label{fig:AlgStrongMsmall}
   \end{figure}

 \begin{proof}
   The pseudocode of the algorithm is shown in Figure~\ref{fig:AlgStrongMsmall}.
   For computing the strongly $M$-consistent refinement we first compute all
   sets $S$ where for all $T\subseteq S$ we have
   $|\,\eval{\query_{\vsets}\langle T\rangle}{\strucB}|\leq M$; as in
   \cite{Marx.2013}, we say that such sets $S$ are \emph{$M$-small}.
   First note that the empty set is $M$-small.
   For nonempty sets $S$ we know
   that $S$ is only $M$-small if for every $x\in S$ the set
   $S\setminus \{x\}$ is $M$-small and hence already included in
   $\vsets$.
   If this is the case, then 
   $\eval{\query_{\vsets}\langle S\rangle}{\strucB}$ 
   can be computed in time $O(M\cdot n)$ by testing
   for every $h\in \mapR^\strucB_{S\setminus
     \{x\}}$ (for an arbitrary $x\in S$) and every 
   element $c$ in the structure's universe, whether
   $h\cup \{(x,c)\}\in \eval{\query_{\vsets}\langle
   S\rangle}{\strucB}$.
   If $|\eval{\query_{\vsets}\langle
   S\rangle}{\strucB}| \leq M$, then we include $S$ and $\mapR^\strucB_S:=\eval{\query_{\vsets}\langle
   S\rangle}{\strucB}$ into our current refinement.
   Afterwards, we want to satisfy the condition on strong $M$-consistency \eqref{eq:stronglyMcons}  by trying all
   pairs of $M$-small sets $S$ and $T$. This is the bottleneck of our
   algorithm and requires time $O(|\mapR^\strucB_S|\cdot |\mapR^\strucB_T|) \leq O(M^2)$.
   In the third step we apply Lemma~\ref{lem:consistent} to enforce consistency of the
   current refinement. In particular, every set $S\cup T$ that
   was found in step~2 becomes $M$-small.
   Note that after deleting tuples to ensure consistency, new sets
   may become $M$-small. Therefore, we have to repeat steps
   1--3 until no more sets became $M$-small.
   Overall, we repeat the outer loop at most $2^{k}$ times,
   step~1 takes time $2^{O(k)}\cdot M \cdot n$, step~2 takes time $2^{O(k)}\cdot
   M^2$ and step~3 takes time $2^{O(k)}\cdot M$. Since $n\leq M$ this leads
   to the required running time.
 \end{proof}

The key step in the proof of Theorem~\ref{thm:fcsubw-with-testing} is
to compute $f(\query)$ strongly $M$-consistent refinements $(\vsets_i,\strucB_i)$ of
$\query$ and $\strucA$ such that $\eval{\query}{\strucA}=\bigcup_{i}
\eval{\query_{\vsets_i}}{\strucB_i}$. In addition to being
strongly $M$-consistent, we want the structures $\strucB_i$ to be
uniform in the sense that the degree of tuples (i.\,e. the number of
extensions) is roughly the average degree.
We make this precise in a moment, but for illustration it might be
helpful to consult 
the example from Section~\ref{subsection:subw}
again.
In every relation in $\strucA$ there is one vertex ($a$ or $b$) of
out-degree $\ell$ and there are $\ell$ vertices of out-degree $1$. Hence the
average out-degree is  $2\ell/(\ell+1)$ and the vertex degrees are highly imbalanced.
However, after splitting the instance in $\mathcal A'$ and $\mathcal
A''$, in every relation, all vertices have either out-degree $\ell$ or
$1$ and the out-degree of every vertex matches the average out-degree
of the corresponding relation.
The next definition generalises this to tuples of variables.
We call a refinement $(\vsets,\strucB)$ \emph{non-trivial}, if every
additional relation in the expansion $\strucB$ is non-empty.
For a non-trivial consistent refinement $(\vsets,\strucB)$ and
$S,T\in\vsets$, $S\subseteq T$, we let

\newcommand{\maxdeg}{\ensuremath{\operatorname{maxdeg}}}
\newcommand{\avgdeg}{\ensuremath{\operatorname{avgdeg}}}

\begin{align}
  \label{eq:8}
  \avgdeg(S,T) &:= |\mapR^{\strucB}_{
                 T}|/|\mapR^{\strucB}_{S}| \qquad\text{and}\\
  \maxdeg(S,T) &:= \max_{g\in \mapR^\strucB_S}\Setc{h\in \mapR^\strucB_T}{\pi_S(h) = g}.
\end{align}

Note that consistency ensures that these numbers are well-defined and
non-zero. 
Furthermore, we can compute them from $(\vsets,\strucB)$ in
time $O(|\vsets|^2\cdot\|\strucB\|)$.
By definition we have $\maxdeg(S,T)\geq\avgdeg(S,T)$. The next
definition states that maximum degree does not deviate too much from
the average degree.

\begin{definition}
  Let $(\vsets,\strucB)$ be a non-trivial consistent refinement of
  $\query$ and $\mathcal A$,
  and let $m$ be the number of tuples of largest
relation of $\strucA$. The
  refinement $(\vsets,\strucB)$ is \emph{$\varepsilon$-uniform} if 
for all $S,T\in \vsets$ with $S\subseteq T$ we have
$\maxdeg(S,T)\leq m^\varepsilon\cdot\avgdeg(S,T)$.
\end{definition}

The next lemma uses Marx's splitting routine to obtain
a partition into  strongly
   $M$-consistent $\varepsilon$-uniform refinements, for $M\deff m^c$.

\begin{lemma}\label{lem:decompdisjoint}
   Let $\query$ be a quantifier-free CQ, let $\strucA$ be a
   $\sigma(\query)$-structure where the largest relation contains $m$
   tuples, and let $c\geq 
   1$ and $\epsilon>0$ be real numbers.
   There is 
   a computable function $f$ and
   an algorithm that computes in time
   $O(f(\query,c,\varepsilon)\cdot m^{2c})$ and
   space $O(f(\query,c,\varepsilon)\cdot m^c)$ a
   sequence of $\ell\leq f(\query,c,\epsilon)$ 
   strongly
   $m^c$-consistent $\varepsilon$-uniform refinements $(\vsets_i,\strucB_i)$ such that
   $\eval{\query}{\strucA}$ is the disjoint union of
   the sets $\eval{\query_{\vsets_i}}{\strucB_i}$.
\end{lemma}

\begin{proof}[Proof (sketch)]
  We follow the same splitting strategy as in \cite{Marx.2013}, but use the improved
  algorithm from Lemma~\ref{lem:strongMcons} to ensure strong $m^c$-consistency.
  Starting with the trivial refinement $(\emptyset,\mathcal A)$, in
  each step we first apply Lemma~\ref{lem:strongMcons} to ensure strong
  $m^c$-consistency. Afterwards, we check whether the current
  refinement $(\vsets,\strucB)$ contains sets
  $S,T\in \vsets$ that contradict $\epsilon$-uniformity, i.\,e.,
  $S\subseteq T$ and
    $\maxdeg(S,T) > m^\varepsilon\cdot\avgdeg(S,T)$.
  If this is the case, we split the refinement $(\vsets,\strucB)$ into
  $(\vsets,\strucB')$ and $(\vsets,\strucB'')$ such that
  $\mapR^\strucB_S$ is partitioned into tuples of small 
  degree and tuples of large degree:
  \begin{align}
    \label{eq:10}
    \mapR^{\strucB'}_{U} &= \mapR^{\strucB''}_{U} :=
                           \mapR^{\strucB}_{U} \quad\text{for all
                                                  $U\in\vsets\setminus\{S\}$,
                                                  }\\
    \mapR^{\strucB'}_{S} &:= \Setc{g\in
                           \mapR^\strucB_S}{\big|{\Setc{h\in
                           \mapR^\strucB_T}{\pi_S(h) = g}}\big| \leq
                           m^{\varepsilon/2}\cdot\avgdeg(S,T)} \\
    \mapR^{\strucB''}_{S} &:= \Setc{g\in
                           \mapR^\strucB_S}{\big|{\Setc{h\in
                           \mapR^\strucB_T}{\pi_S(h) = g}}\big| >
                           m^{\varepsilon/2}\cdot\avgdeg(S,T)}
  \end{align}
  It is clear that $\eval{\query}{\strucB}$ is the disjoint union of
  $\eval{\query}{\strucB'}$ and $\eval{\query}{\strucB''}$ and that
  the recursion terminates at some point with a sequence of strongly
  $m^c$-consistent $\epsilon$-uniform refinements that partition
  $\eval{\query}{\strucA}$.
  It is also not hard to show that the height of the recursion tree
  is bounded by $2^{O(|\vars(\query)|)}\cdot \frac{c}{\epsilon}$ (see
  \cite[Lemma~4.11]{Marx.2013}). Hence, by Lemma~\ref{lem:strongMcons} the procedure
  can be implemented in time
   $O(f(\query,c,\epsilon)\cdot m^{2c})$ and
   space $O(f(\query,c,\epsilon)\cdot m^c)$.
\end{proof}

The nice thing about $\epsilon$-uniform and
strongly $m^c$-consistent refinements is that they define, for small enough
$\epsilon$, a submodular function $g\in \SubWidthSet(\query)$, which
in turn guarantees the existence of a tree decomposition with small
projections on the bags.
The following lemma from \cite[Lemma~4.12]{Marx.2013} provides these functions. However,
there is an oversight in Marx's proof and in order to fix this, we have
to ensure \emph{strong} $m^c$-consistency instead of only
$m^c$-consistency as 
stated in \cite[Lemma~4.12]{Marx.2013}.
As suggested by Marx (personal communication), an alternative way to achieve strong $m^c$-consistency would be to enforce
$m^{2c}$-consistency, which leads to the same runtime guarantees, but
requires more space.

 \begin{lemma}\label{lem:defsubwfunc}
   Let $(\vsets,\strucB)$ be an $\epsilon$-uniform
   strongly $m^c$-consistent refinement of $\query$ and $\strucA$, and
   let $c\geq
   1$ and $|\vars(\query)|^{-3}\geq\epsilon>0$ be real numbers.
   Then $g_{\vsets,\strucB}\colon 2^{\Vars(\query)} \to \RRpos$ is a monotone, edge-dominated,
   submodular function that
satisfies $g_{\vsets,\strucB}(\emptyset) = 0$:
\begin{align}
  \label{eq:11}
  g_{\vsets,\strucB}(U) &:=
         \begin{cases}
           (1-\varepsilon^{1/3})\cdot
           \log_m\big(|\mapR^\strucB_U|\big) + h(U) &\text{if $U\in\vsets$}\\
                      (1-\varepsilon^{1/3})\cdot
           c + h(U) &\text{if $U\notin\vsets$,}
         \end{cases}
\end{align}
where $h(U) := 2\epsilon^{2/3}|U| -
           \epsilon|U|^2 \geq 0$ for all $U\subseteq\vars(\query)$.
\end{lemma}

The proof can be copied verbatim from Marx's proof of
\cite[Lemma~4.12]{Marx.2013} by using the notion of strong consistency
instead of plain consistency.
For the reader's convenience, we
provide the proof below.

\begin{proof}[Proof of Lemma~\ref{lem:defsubwfunc} (Lemma 4.12 in \cite{Marx.2013})]
  The function $h$ is non-negative and monotone in the range  $0\leq
  |U|\leq 1/\epsilon^{1/3}$. In particular, $0\leq h(S) \leq h(T)\leq \varepsilon^{1/3}$ for
  all $S\subseteq T\subseteq \vars(\query)$.
  Moreover $h$ is submodular:
  \begin{equation}
    \label{eq:12}
  h(S) + h(T) - h(S\cap T)  - h(S\cup T) = 2\varepsilon\cdot |S\setminus
  T|\cdot |T\setminus S| \geq 0\;.
  \end{equation}

  The monotonicity of $g_{\vsets,\strucB}$ follows from the monotonicity of $h$ and the
  $m^c$-consistency of the refinement. To see that
  $g_{\vsets,\strucB}$ is edge-dominated, note that $\vars(\alpha)$ is
  $m^c$-consistent for every $c\geq 1$ and every
  $\qatom\in\atoms(\query)$. Hence, $g_{\vsets,\strucB}(\vars(\alpha)) \leq
  (1-\varepsilon^{1/3})+h(\vars(\alpha))\leq 1$.
  
  Now we have to verify the submodularity condition
  \begin{equation}
    \label{eq:13}
  g_{\vsets,\strucB}(S) +
  g_{\vsets,\strucB}(T) - g_{\vsets,\strucB}(S\cap T)  -
  g_{\vsets,\strucB}(S\cup T) \geq 0.
  \end{equation}
  This is trivial when $S\subseteq T$ or
  $T\subseteq S$.
  Thus we can assume that $|S\setminus
  T|\geq 1$ and  $|T\setminus S|\geq 1$, which by \eqref{eq:12}
implies 
\begin{equation}\label{eq:CQ}\tag{$\ast$}
  h(S) + h(T) - h(S\cap T)  - h(S\cup T) \geq 2\varepsilon.
\end{equation}
  If at least one of $S$ and $T$ is not contained in  $\vsets$, then
  \eqref{eq:13} follows from $\log_m\big(|\mapR^\strucB_U|\big)\leq c$
  and the submodularity of $h$.
  The remaining case is that $S \in \vsets$ and $T \in \vsets$. Here
  we have
  \begin{align}
    \label{eq:14}
      &g_{\vsets,\strucB}(S) +
  g_{\vsets,\strucB}(T)\\
    &=
                          (1-\varepsilon^{1/3})\cdot
           \log_m\big(|\mapR^\strucB_S|\big)
+                          (1-\varepsilon^{1/3})\cdot
           \log_m\big(|\mapR^\strucB_T|\big)
                          + h(S)
                          + h(T) \\
&=
                          (1-\varepsilon^{1/3})\cdot
           \log_m\big(|\mapR^\strucB_S|\big)
+                          (1-\varepsilon^{1/3})\cdot
           \log_m\big(|\mapR^\strucB_{S\cap T}|\cdot\avgdeg(S\cap T, T)\big)
                          \\
                          &\quad+ h(S)
                          + h(T) \\
      &\geq
                          (1-\varepsilon^{1/3})\cdot
           \log_m\big(|\mapR^\strucB_S|\big)
+                          (1-\varepsilon^{1/3})\cdot
           \log_m\big(|\mapR^\strucB_{S\cap T}|\big)\\
    &\quad+(1-\varepsilon^{1/3})\cdot\log_m\big(\maxdeg(S\cap T, T)/m^\epsilon\big)
+ h(S)
                          + h(T) \\
      &=
                          (1-\varepsilon^{1/3})\cdot
           \log_m\big(|\mapR^\strucB_{S\cap T}|\big)
    +(1-\varepsilon^{1/3})\cdot\log_m\big(|\mapR^\strucB_S|\cdot\maxdeg(S\cap
        T, T)\big)\\
    &\quad
-(1-\varepsilon^{1/3})\varepsilon + h(S)
                          + h(T) \\
      &\geq
                          (1-\varepsilon^{1/3})\cdot
           \log_m\big(|\mapR^\strucB_{S\cap T}|\big)
    +(1-\varepsilon^{1/3})\cdot\log_m\big(|\mapR^\strucB_S|\cdot\maxdeg(S,
        S\cup T)\big)\\
    &\quad
-(1-\varepsilon^{1/3})\varepsilon + h(S\cap T)  + h(S\cup T) + 2\varepsilon \\
      &\geq
                          (1-\varepsilon^{1/3})\cdot
           \log_m\big(|\mapR^\strucB_{S\cap T}|\big)
    +(1-\varepsilon^{1/3})\cdot\log_m\big(|\mapR^\strucB_{S\cup T}|\big)\\
    &\quad
+ h(S\cap T)  + h(S\cup T) \\
    &\geq g_{\vsets,\strucB}(S\cap T) +
  g_{\vsets,\strucB}(S\cup T)
  \end{align}

The first inequality holds because of $\varepsilon$-uniformity. The
second inequality holds, because in general $\maxdeg(X, Y) \geq
\maxdeg(X\cup Z, Y\cup Z)$ and ($\ast$). The last inequality holds
because $S\cap T \in \vsets$ by consistency and because of strong
$m^c$-consistency  we have either $|\mapR^\strucB_{S\cup T}| > m^c$ or
$S\cup T\in\vsets$ (and this is where the new requirement of strong $m^c$-consistency is needed).
\end{proof}

Now we are ready to prove our main theorem.

\newcommand{\qfquery}{\widetilde{\query}}
\newcommand{\qfpsi}{\widetilde{\psi}}

\begin{proof}[Proof of Theorem~\ref{thm:fcsubw-with-testing}]
  We fix  $c=(1+\delta)\widthw$ and let $\varepsilon$ be the minimum of
  $\big(1-1/(1+\delta)\big)^4$ and
  $|\vars(\query)|^{-4}$.
  Suppose that 
 $\query$ is of the form $\exists x_1\cdots\exists x_k\, \qfquery$ 
 where 
  $\qfquery$ is quantifier-free.
  We apply Lemma~\ref{lem:decompdisjoint} to $\qfquery$, $\strucA$,
  $c$, $\varepsilon$ to obtain in time $O(f(\query) m^{2c})$ a
  sequence of $\ell\leq f(\query)$
  strongly
  $m^c$-consistent $\varepsilon$-uniform refinements
  $(\vsets_i,\strucB_i)$ such that
   $\eval{\qfquery}{\strucA}$ is the disjoint union of
   $\eval{\qfquery_{\vsets_1}}{\strucB_1}$, \ldots, $\eval{\qfquery_{\vsets_\ell}}{\strucB_\ell}$.
  By Lemma~\ref{lem:defsubwfunc} we have  $g_{\vsets_i,\strucB_i}\in
  \SubWidthSet(\qfquery)=\SubWidthSet(\query)$ for every $i\in[\ell]$. 
  Hence, by the definition of free-connex submodular width 
  \eqref{eq:defsubw}, we know that there is a free-connex
  tree decomposition $(\Tree_i,\Bag_i)%
  $ of $\query$
  such that $g_{\vsets_i,\strucB_i}(\Bag_i(\treenode)) \leq \widthw$
  for every $t\in\Nodes(\Tree_i)$.
  Note that by the choice of 
  $c$ , $\epsilon$ and the
  non-negativity of $h$ (see Lemma~\ref{lem:defsubwfunc}) we have
  \begin{align}
    \label{eq:5}
    w = c / (1+\delta) \leq (1-\varepsilon^{1/4}) \cdot c <
    (1-\varepsilon^{1/3}) \cdot c + h(U).
  \end{align}
  Hence,  $g_{\vsets_i,\strucB_i}(U) \leq w$ implies $U\in\vsets$ and
  therefore $|\mapR^{\strucB_i}_U|=|\,\eval{\query_{\vsets_i}\langle
  U\rangle}{\strucB_i}|\leq m^c$ by \eqref{eq:cons1} and
  \eqref{eq:Mcons}.
  Thus, every bag of the free-connex tree-decomposition $(\Tree_i,\Bag_i)$ is small in
  the $i$th refinement. However, $(\Tree_i,\Bag_i)$ is a
  tree-decomposition of $\query$, but not necessarily
  of $\query_{\vsets_i}$! In fact, $\query_{\vsets_i}$ can be very
  dense, e.\,g., if $\vsets_i=2^{\vars(\query)}$. To take care of
  this, we thin out the refinement and only keep those atoms and
  relations that correspond to bags of the decomposition.
  In particular,
  for every $i\in[\ell]$
  we define
    $\qfpsi_i := \textstyle\bigwedge_{t\in\Nodes(\Tree_i)}
             R_{\Bag_i(\treenode)}(x_{\Bag_i(\treenode)})$ and let
   $\psi_i := \exists x_1\cdots\exists x_k\, \qfpsi_i$
     be the quantified version.
    Note that $\psi_i$ is a free-connex acyclic CQ.
    Additionally, we %
    let 
    $\strucC_i$  be the $\sigma(\psi_i)$-reduct of $\strucB_i$.
    We argue that 
    $\eval{\qfquery_{\vsets_i}}{\strucB_i}\subseteq
    \eval{\qfpsi_i}{\strucC_i}\subseteq \eval{\qfquery}{\strucA}$. The
    first inclusion holds because $\qfquery_{\vsets_i}$ and
    $\strucB_i$ refine $\qfpsi_i$ and $\strucC_i$. The second
    inclusion holds because every atom from  $\qfquery$ is contained
    in a bag of the decomposition and is hence covered by
    an atom in $\qfpsi_i$ because of consistency.
    It therefore also follows that 
$\pi_F\big(\eval{\qfquery_{\vsets_i}}{\strucB_i}\big)\subseteq
    \pi_F\big(\eval{\qfpsi_i}{\strucC_i}\big)\subseteq
    \pi_F\big(\eval{\qfquery}{\strucA}\big)$ for $F\deff \free(\phi)$,
 and
    hence $\eval{\query_{\vsets_i}}{\strucB_i}\subseteq
    \eval{\psi_i}{\strucC_i}\subseteq \eval{\query}{\strucA}$.
    Overall, we have that
    $\eval{\query}{\strucA}=\bigcup_{i\in[\ell]}\eval{\psi_i}{\strucC_i}$,
    where the union is not necessarily disjoint, each $\psi_i$ is
    free-connex acyclic, and $\|\strucC_i\| =
    O(|\vars(\query)|^2m^{(1+\delta)\widthw})$. By combining 
   Theorem~\ref{thm:BaganEtAlRefined}
   with 
   Theorem~\ref{thm:UnionTrick},
   the theorem follows.
\end{proof}

\section{Final Remarks}\label{section:conclusion}\label{section:FinalRemarks}

In this paper,
we have investigated the enumeration complexity of conjunctive queries and
have shown that 
every class of conjunctive queries of bounded
free-connex submodular width admits constant delay enumeration with
FPT-preprocessing.
These are by now the largest classes of CQs that allow efficient
  enumeration
in this sense.

For quantifier-free self-join-free CQs this upper bound
is matched by Marx's lower bound
\cite{Marx.2013}.
I.\,e., 
recursively enumerable classes of quantifier-free
self-join-free CQs of unbounded free-connex submodular width do not
admit constant delay enumeration after FPT-preprocessing (assuming
the exponential time hypothesis 
ETH).

A major future  
task is to obtain a complete dichotomy, or at
least one for all self-join-free CQs.
The gray-zone 
for the latter
are classes of CQs that have bounded submodular width,
but unbounded free-connex submodular width.
An intriguing  example in this gray-zone 
is the $k$-star query with a quantified center,
i.\,e., the query $\psi_k$ of the form
$\exists z\,\textstyle\bigwedge^k_{i=1} R_i(z,x_i)$.
  Here we have $\SUBW(\psi_k)=1$ and $\fcSUBW(\psi_k)=k$.
  It is open whether the class $\Psi=\setc{\psi_k}{k\in\NNpos}$
  admits constant delay enumeration with FPT-preprocessing.

  \medskip
  \noindent
  \textbf{Acknowledgements}  Funded by the German Research Foundation (Deutsche
 Forschungsgemeinschaft, DFG) --  project numbers
 316451603; 414325841
 (gef\"ordert durch die Deutsche Forschungsgemeinschaft (DFG) --
 Projektnummern 316451603; 414325841).

  \bibliographystyle{plainurl}%

\bibliography{literature}

\begin{thebibliography}{10}

\bibitem{AHV-Book}
Serge Abiteboul, Richard Hull, and Victor Vianu.
\newblock {\em Foundations of Databases}.
\newblock Addison-Wesley, 1995.
\newblock URL: \url{http://webdam.inria.fr/Alice/}.

\bibitem{Khamis.2017}
Mahmoud Abo~Khamis, Hung~Q. Ngo, and Dan Suciu.
\newblock What do {S}hannon-type inequalities, submodular width, and
  disjunctive datalog have to do with one another?
\newblock In {\em Proceedings of the 36th {ACM} {SIGMOD-SIGACT-SIGAI} Symposium
  on Principles of Database Systems {PODS} 2017}, pages 429--444, 2017.
\newblock Full version available at CoRR, abs/1612.02503, 2016 (URL:
  \url{http://arxiv.org/abs/1612.02503}).
\newblock URL: \url{http://doi.acm.org/10.1145/3034786.3056105}, \href
  {http://dx.doi.org/10.1145/3034786.3056105}
  {\path{doi:10.1145/3034786.3056105}}.

\bibitem{Alon.1997}
Noga Alon, Raphael Yuster, and Uri Zwick.
\newblock Finding and counting given length cycles.
\newblock {\em Algorithmica}, 17(3):209--223, 1997.
\newblock URL: \url{https://doi.org/10.1007/BF02523189}, \href
  {http://dx.doi.org/10.1007/BF02523189} {\path{doi:10.1007/BF02523189}}.

\bibitem{DBLP:conf/icdt/AmarilliBM18}
Antoine Amarilli, Pierre Bourhis, and Stefan Mengel.
\newblock Enumeration on trees under relabelings.
\newblock In {\em 21st International Conference on Database Theory, {ICDT}
  2018, March 26-29, 2018, Vienna, Austria}, pages 5:1--5:18, 2018.
\newblock URL: \url{https://doi.org/10.4230/LIPIcs.ICDT.2018.5}, \href
  {http://dx.doi.org/10.4230/LIPIcs.ICDT.2018.5}
  {\path{doi:10.4230/LIPIcs.ICDT.2018.5}}.

\bibitem{DBLP:conf/icdt/AmarilliBMN19}
Antoine Amarilli, Pierre Bourhis, Stefan Mengel, and Matthias Niewerth.
\newblock Constant-delay enumeration for nondeterministic document spanners.
\newblock In {\em 22nd International Conference on Database Theory, {ICDT}
  2019, March 26-28, 2019, Lisbon, Portugal}, pages 22:1--22:19, 2019.
\newblock URL: \url{https://doi.org/10.4230/LIPIcs.ICDT.2019.22}, \href
  {http://dx.doi.org/10.4230/LIPIcs.ICDT.2019.22}
  {\path{doi:10.4230/LIPIcs.ICDT.2019.22}}.

\bibitem{DBLP:conf/pods/AmarilliBMN19}
Antoine Amarilli, Pierre Bourhis, Stefan Mengel, and Matthias Niewerth.
\newblock Enumeration on trees with tractable combined complexity and efficient
  updates.
\newblock In {\em Proceedings of the 38th {ACM} {SIGMOD-SIGACT-SIGAI} Symposium
  on Principles of Database Systems, {PODS} 2019, Amsterdam, The Netherlands,
  June 30 -- July 5, 2019}, pages 89--103, 2019.
\newblock URL: \url{https://doi.org/10.1145/3294052.3319702}, \href
  {http://dx.doi.org/10.1145/3294052.3319702}
  {\path{doi:10.1145/3294052.3319702}}.

\bibitem{DBLP:conf/csl/Bagan06}
Guillaume Bagan.
\newblock {MSO} queries on tree decomposable structures are computable with
  linear delay.
\newblock In {\em Computer Science Logic, 20th International Workshop, {CSL}
  2006, 15th Annual Conference of the EACSL, Szeged, Hungary, September 25-29,
  2006, Proceedings}, pages 167--181, 2006.
\newblock URL: \url{https://doi.org/10.1007/11874683\_11}, \href
  {http://dx.doi.org/10.1007/11874683\_11} {\path{doi:10.1007/11874683\_11}}.

\bibitem{Bagan_PhD}
Guillaume Bagan.
\newblock {\em Algorithmes et complexit{\'{e}} des probl{\`{e}}mes
  d'{\'{e}}num{\'{e}}ration pour l'{\'{e}}valuation de requ{\^{e}}tes logiques.
  (Algorithms and complexity of enumeration problems for the evaluation of
  logical queries)}.
\newblock PhD thesis, University of Caen Normandy, France, 2009.
\newblock URL: \url{https://tel.archives-ouvertes.fr/tel-00424232}.

\bibitem{Bagan.2007}
Guillaume Bagan, Arnaud Durand, and Etienne Grandjean.
\newblock On acyclic conjunctive queries and constant delay enumeration.
\newblock In {\em Proceedings of the 16th Annual Conference of the EACSL,
  CSL'07, Lausanne, Switzerland, September 11--15, 2007}, pages 208--222, 2007.
\newblock URL: \url{http://dx.doi.org/10.1007/978-3-540-74915-8_18}, \href
  {http://dx.doi.org/10.1007/978-3-540-74915-8_18}
  {\path{doi:10.1007/978-3-540-74915-8_18}}.

\bibitem{DBLP:journals/jacm/BeeriFMY83}
Catriel Beeri, Ronald Fagin, David Maier, and Mihalis Yannakakis.
\newblock On the desirability of acyclic database schemes.
\newblock {\em J. {ACM}}, 30(3):479--513, 1983.
\newblock URL: \url{http://doi.acm.org/10.1145/2402.322389}, \href
  {http://dx.doi.org/10.1145/2402.322389} {\path{doi:10.1145/2402.322389}}.

\bibitem{SiglogNewsTutorial_BGS2020}
Christoph Berkholz, Fabian Gerhardt, and Nicole Schweikardt.
\newblock Constant delay enumeration for conjunctive queries --- a tutorial.
\newblock {\em {SIGLOG} News}, 7(1):4--33, 2020.
\newblock URL: \url{https://doi.org/10.1145/3385634.3385636}.

\bibitem{BKS_enumeration_PODS17}
Christoph Berkholz, Jens Keppeler, and Nicole Schweikardt.
\newblock Answering conjunctive queries under updates.
\newblock In {\em Proceedings of the 36th {ACM} {SIGMOD-SIGACT-SIGAI} Symposium
  on Principles of Database Systems, PODS'17, Chicago, IL, USA, May 14--19,
  2017}, pages 303--318, 2017.
\newblock Full version available at \url{http://arxiv.org/abs/1702.06370}.
\newblock URL: \url{http://doi.org/10.1145/3034786.3034789}, \href
  {http://dx.doi.org/10.1145/3034786.3034789}
  {\path{doi:10.1145/3034786.3034789}}.

\bibitem{BKS-ICDT17}
Christoph Berkholz, Jens Keppeler, and Nicole Schweikardt.
\newblock Answering {FO+MOD} queries under updates on bounded degree databases.
\newblock {\em {ACM} Trans. Database Syst.}, 43(2):7:1--7:32, 2018.
\newblock URL: \url{https://doi.org/10.1145/3232056}, \href
  {http://dx.doi.org/10.1145/3232056} {\path{doi:10.1145/3232056}}.

\bibitem{DBLP:conf/icdt/BerkholzKS18}
Christoph Berkholz, Jens Keppeler, and Nicole Schweikardt.
\newblock Answering {UCQs} under updates and in the presence of integrity
  constraints.
\newblock In {\em 21st International Conference on Database Theory, {ICDT}
  2018, March 26-29, 2018, Vienna, Austria}, pages 8:1--8:19, 2018.
\newblock URL: \url{https://doi.org/10.4230/LIPIcs.ICDT.2018.8}, \href
  {http://dx.doi.org/10.4230/LIPIcs.ICDT.2018.8}
  {\path{doi:10.4230/LIPIcs.ICDT.2018.8}}.

\bibitem{DBLP:conf/mfcs/BerkholzS19}
Christoph Berkholz and Nicole Schweikardt.
\newblock Constant delay enumeration with {FPT}-preprocessing for conjunctive
  queries of bounded submodular width.
\newblock In {\em 44th International Symposium on Mathematical Foundations of
  Computer Science, {MFCS} 2019, August 26--30, 2019, Aachen, Germany}, pages
  58:1--58:15, 2019.
\newblock URL: \url{https://doi.org/10.4230/LIPIcs.MFCS.2019.58}, \href
  {http://dx.doi.org/10.4230/LIPIcs.MFCS.2019.58}
  {\path{doi:10.4230/LIPIcs.MFCS.2019.58}}.

\bibitem{DBLP:journals/siamcomp/BernsteinG81}
Philip~A. Bernstein and Nathan Goodman.
\newblock Power of natural semijoins.
\newblock {\em {SIAM} J. Comput.}, 10(4):751--771, 1981.
\newblock URL: \url{https://doi.org/10.1137/0210059}, \href
  {http://dx.doi.org/10.1137/0210059} {\path{doi:10.1137/0210059}}.

\bibitem{BraultBaron_PhD}
Johann Brault{-}Baron.
\newblock {\em De la pertinence de l'{\'{e}}num{\'{e}}ration : complexit{\'{e}}
  en logiques propositionnelle et du premier ordre. (The relevance of the list:
  propositional logic and complexity of the first order)}.
\newblock PhD thesis, University of Caen Normandy, France, 2013.
\newblock URL: \url{https://tel.archives-ouvertes.fr/tel-01081392}.

\bibitem{DBLP:journals/csur/Brault-Baron16}
Johann Brault{-}Baron.
\newblock Hypergraph acyclicity revisited.
\newblock {\em {ACM} Comput. Surv.}, 49(3):54:1--54:26, 2016.
\newblock URL: \url{http://doi.acm.org/10.1145/2983573}, \href
  {http://dx.doi.org/10.1145/2983573} {\path{doi:10.1145/2983573}}.

\bibitem{DBLP:conf/pods/DeepK18}
Shaleen Deep and Paraschos Koutris.
\newblock Compressed representations of conjunctive query results.
\newblock In {\em Proceedings of the 37th {ACM} {SIGMOD-SIGACT-SIGAI} Symposium
  on Principles of Database Systems, Houston, TX, USA, June 10-15, 2018}, pages
  307--322, 2018.
\newblock URL: \url{http://doi.acm.org/10.1145/3196959.3196979}, \href
  {http://dx.doi.org/10.1145/3196959.3196979}
  {\path{doi:10.1145/3196959.3196979}}.

\bibitem{DBLP:journals/jlp/DowlingG84}
William~F. Dowling and Jean~H. Gallier.
\newblock Linear-time algorithms for testing the satisfiability of
  propositional horn formulae.
\newblock {\em J. Log. Program.}, 1(3):267--284, 1984.
\newblock URL: \url{https://doi.org/10.1016/0743-1066(84)90014-1}, \href
  {http://dx.doi.org/10.1016/0743-1066(84)90014-1}
  {\path{doi:10.1016/0743-1066(84)90014-1}}.

\bibitem{DurandGrandjean_BoundedDegree}
Arnaud Durand and Etienne Grandjean.
\newblock First-order queries on structures of bounded degree are computable
  with constant delay.
\newblock {\em {ACM} Trans. Comput. Log.}, 8(4), 2007.
\newblock \href {http://dx.doi.org/10.1145/1276920.1276923}
  {\path{doi:10.1145/1276920.1276923}}.

\bibitem{DBLP:conf/pods/DurandSS14}
Arnaud Durand, Nicole Schweikardt, and Luc Segoufin.
\newblock Enumerating answers to first-order queries over databases of low
  degree.
\newblock In {\em Proceedings of the 33rd {ACM} {SIGMOD-SIGACT-SIGART}
  Symposium on Principles of Database Systems, PODS'14, Snowbird, UT, USA, June
  22--27, 2014}, pages 121--131, 2014.
\newblock \href {http://dx.doi.org/10.1145/2594538.2594539}
  {\path{doi:10.1145/2594538.2594539}}.

\bibitem{DBLP:conf/csl/DurandS11}
Arnaud Durand and Yann Strozecki.
\newblock Enumeration complexity of logical query problems with second-order
  variables.
\newblock In {\em Computer Science Logic, 25th International Workshop / 20th
  Annual Conference of the EACSL, {CSL} 2011, September 12-15, 2011, Bergen,
  Norway, Proceedings}, pages 189--202, 2011.
\newblock URL: \url{https://doi.org/10.4230/LIPIcs.CSL.2011.189}, \href
  {http://dx.doi.org/10.4230/LIPIcs.CSL.2011.189}
  {\path{doi:10.4230/LIPIcs.CSL.2011.189}}.

\bibitem{GottlobEtAl_QandA}
Georg Gottlob, Gianluigi Greco, Nicola Leone, and Francesco Scarcello.
\newblock Hypertree decompositions: Questions and answers.
\newblock In {\em Proceedings of the 35th {ACM} {SIGMOD-SIGACT-SIGAI} Symposium
  on Principles of Database Systems, {PODS} 2016, San Francisco, CA, USA, June
  26 - July 01, 2016}, pages 57--74, 2016.
\newblock URL: \url{http://doi.acm.org/10.1145/2902251.2902309}, \href
  {http://dx.doi.org/10.1145/2902251.2902309}
  {\path{doi:10.1145/2902251.2902309}}.

\bibitem{DBLP:journals/jcss/GottlobLS02}
Georg Gottlob, Nicola Leone, and Francesco Scarcello.
\newblock Hypertree decompositions and tractable queries.
\newblock {\em J. Comput. Syst. Sci.}, 64(3):579--627, 2002.
\newblock URL: \url{https://doi.org/10.1006/jcss.2001.1809}, \href
  {http://dx.doi.org/10.1006/jcss.2001.1809}
  {\path{doi:10.1006/jcss.2001.1809}}.

\bibitem{DBLP:journals/talg/GroheM14}
Martin Grohe and D{\'{a}}niel Marx.
\newblock Constraint solving via fractional edge covers.
\newblock {\em {ACM} Trans. Algorithms}, 11(1):4:1--4:20, 2014.
\newblock URL: \url{http://doi.acm.org/10.1145/2636918}, \href
  {http://dx.doi.org/10.1145/2636918} {\path{doi:10.1145/2636918}}.

\bibitem{DynamicYannakakis2017}
Muhammad Idris, Mart{\'{\i}}n Ugarte, and Stijn Vansummeren.
\newblock The {Dynamic Yannakakis Algorithm}: Compact and efficient query
  processing under updates.
\newblock In {\em Proceedings of the 2017 {ACM} International Conference on
  Management of Data, {SIGMOD} Conference 2017, Chicago, IL, USA, May 14-19,
  2017}, pages 1259--1274, 2017.
\newblock URL: \url{http://doi.acm.org/10.1145/3035918.3064027}, \href
  {http://dx.doi.org/10.1145/3035918.3064027}
  {\path{doi:10.1145/3035918.3064027}}.

\bibitem{DBLP:conf/icdt/KaraO18}
Ahmet Kara and Dan Olteanu.
\newblock Covers of query results.
\newblock In {\em 21st International Conference on Database Theory, {ICDT}
  2018, March 26-29, 2018, Vienna, Austria}, pages 16:1--16:22, 2018.
\newblock URL: \url{https://doi.org/10.4230/LIPIcs.ICDT.2018.16}, \href
  {http://dx.doi.org/10.4230/LIPIcs.ICDT.2018.16}
  {\path{doi:10.4230/LIPIcs.ICDT.2018.16}}.

\bibitem{KazanaSegoufin_BoundedDegree}
Wojciech Kazana and Luc Segoufin.
\newblock First-order query evaluation on structures of bounded degree.
\newblock {\em Logical Methods in Computer Science}, 7(2), 2011.
\newblock \href {http://dx.doi.org/10.2168/LMCS-7(2:20)2011}
  {\path{doi:10.2168/LMCS-7(2:20)2011}}.

\bibitem{DBLP:conf/pods/KazanaS13}
Wojciech Kazana and Luc Segoufin.
\newblock Enumeration of first-order queries on classes of structures with
  bounded expansion.
\newblock In {\em Proceedings of the 32nd {ACM} {SIGMOD-SIGACT-SIGART}
  Symposium on Principles of Database Systems, {PODS} 2013, New York, NY, {USA}
  - June 22 - 27, 2013}, pages 297--308, 2013.
\newblock URL: \url{http://doi.acm.org/10.1145/2463664.2463667}, \href
  {http://dx.doi.org/10.1145/2463664.2463667}
  {\path{doi:10.1145/2463664.2463667}}.

\bibitem{DBLP:journals/tocl/KazanaS13}
Wojciech Kazana and Luc Segoufin.
\newblock Enumeration of monadic second-order queries on trees.
\newblock {\em {ACM} Trans. Comput. Log.}, 14(4):25:1--25:12, 2013.
\newblock URL: \url{http://doi.acm.org/10.1145/2528928}, \href
  {http://dx.doi.org/10.1145/2528928} {\path{doi:10.1145/2528928}}.

\bibitem{DBLP:conf/lics/KuskeS17}
Dietrich Kuske and Nicole Schweikardt.
\newblock First-order logic with counting.
\newblock In {\em 32nd Annual {ACM/IEEE} Symposium on Logic in Computer
  Science, {LICS} 2017, Reykjavik, Iceland, June 20-23, 2017}, pages 1--12,
  2017.
\newblock URL: \url{https://doi.org/10.1109/LICS.2017.8005133}, \href
  {http://dx.doi.org/10.1109/LICS.2017.8005133}
  {\path{doi:10.1109/LICS.2017.8005133}}.

\bibitem{DBLP:conf/csl/LosemannM14}
Katja Losemann and Wim Martens.
\newblock {MSO} queries on trees: enumerating answers under updates.
\newblock In {\em Joint Meeting of the Twenty-Third {EACSL} Annual Conference
  on Computer Science Logic {(CSL)} and the Twenty-Ninth Annual {ACM/IEEE}
  Symposium on Logic in Computer Science (LICS), {CSL-LICS} '14, Vienna,
  Austria, July 14 - 18, 2014}, pages 67:1--67:10, 2014.
\newblock URL: \url{http://doi.acm.org/10.1145/2603088.2603137}, \href
  {http://dx.doi.org/10.1145/2603088.2603137}
  {\path{doi:10.1145/2603088.2603137}}.

\bibitem{Marx.2011}
D{\'{a}}niel Marx.
\newblock Tractable structures for constraint satisfaction with truth tables.
\newblock {\em Theory Comput. Syst.}, 48(3):444--464, 2011.
\newblock URL: \url{https://doi.org/10.1007/s00224-009-9248-9}, \href
  {http://dx.doi.org/10.1007/s00224-009-9248-9}
  {\path{doi:10.1007/s00224-009-9248-9}}.

\bibitem{Marx.2013}
D{\'{a}}niel Marx.
\newblock Tractable hypergraph properties for constraint satisfaction and
  conjunctive queries.
\newblock {\em Journal of the ACM (JACM), Volume 60, Issue 6, Article No. 42},
  November 2013.
\newblock URL: \url{http://doi.acm.org/10.1145/2535926}, \href
  {http://dx.doi.org/10.1145/2535926} {\path{doi:10.1145/2535926}}.

\bibitem{DBLP:conf/lics/Niewerth18}
Matthias Niewerth.
\newblock {MSO} queries on trees: Enumerating answers under updates using
  forest algebras.
\newblock In {\em Proceedings of the 33rd Annual {ACM/IEEE} Symposium on Logic
  in Computer Science, {LICS} 2018, Oxford, UK, July 09-12, 2018}, pages
  769--778, 2018.
\newblock URL: \url{http://doi.acm.org/10.1145/3209108.3209144}, \href
  {http://dx.doi.org/10.1145/3209108.3209144}
  {\path{doi:10.1145/3209108.3209144}}.

\bibitem{DBLP:conf/pods/NiewerthS18}
Matthias Niewerth and Luc Segoufin.
\newblock Enumeration of {MSO} queries on strings with constant delay and
  logarithmic updates.
\newblock In {\em Proceedings of the 37th {ACM} {SIGMOD-SIGACT-SIGAI} Symposium
  on Principles of Database Systems, Houston, TX, USA, June 10-15, 2018}, pages
  179--191, 2018.
\newblock URL: \url{http://doi.acm.org/10.1145/3196959.3196961}, \href
  {http://dx.doi.org/10.1145/3196959.3196961}
  {\path{doi:10.1145/3196959.3196961}}.

\bibitem{DBLP:journals/sigmod/OlteanuS16}
Dan Olteanu and Maximilian Schleich.
\newblock Factorized databases.
\newblock {\em {SIGMOD} Record}, 45(2):5--16, 2016.
\newblock URL: \url{http://doi.acm.org/10.1145/3003665.3003667}, \href
  {http://dx.doi.org/10.1145/3003665.3003667}
  {\path{doi:10.1145/3003665.3003667}}.

\bibitem{DBLP:journals/tods/OlteanuZ15}
Dan Olteanu and Jakub Z{\'{a}}vodn{\'{y}}.
\newblock Size bounds for factorised representations of query results.
\newblock {\em {ACM} Trans. Database Syst.}, 40(1):2:1--2:44, 2015.
\newblock URL: \url{http://doi.acm.org/10.1145/2656335}, \href
  {http://dx.doi.org/10.1145/2656335} {\path{doi:10.1145/2656335}}.

\bibitem{Scarcello.2018}
Francesco Scarcello.
\newblock From hypertree width to submodular width and data-dependent
  structural decompositions.
\newblock In {\em Proceedings of the 26th Italian Symposium on Advanced
  Database Systems, Castellaneta Marina (Taranto), Italy, June 24-27, 2018.},
  2018.
\newblock URL: \url{http://ceur-ws.org/Vol-2161/paper24.pdf}.

\bibitem{DBLP:conf/pods/SchweikardtSV18}
Nicole Schweikardt, Luc Segoufin, and Alexandre Vigny.
\newblock Enumeration for {FO} queries over nowhere dense graphs.
\newblock In {\em Proceedings of the 37th {ACM} {SIGMOD-SIGACT-SIGAI} Symposium
  on Principles of Database Systems, Houston, TX, USA, June 10-15, 2018}, pages
  151--163, 2018.
\newblock URL: \url{http://doi.acm.org/10.1145/3196959.3196971}, \href
  {http://dx.doi.org/10.1145/3196959.3196971}
  {\path{doi:10.1145/3196959.3196971}}.

\bibitem{DBLP:journals/sigmod/Segoufin15}
Luc Segoufin.
\newblock Constant delay enumeration for conjunctive queries.
\newblock {\em {SIGMOD} Record}, 44(1):10--17, 2015.
\newblock URL: \url{http://doi.acm.org/10.1145/2783888.2783894}, \href
  {http://dx.doi.org/10.1145/2783888.2783894}
  {\path{doi:10.1145/2783888.2783894}}.

\bibitem{DBLP:conf/icdt/SegoufinV17}
Luc Segoufin and Alexandre Vigny.
\newblock Constant delay enumeration for {FO} queries over databases with local
  bounded expansion.
\newblock In {\em 20th International Conference on Database Theory, {ICDT}
  2017, March 21--24, 2017, Venice, Italy}, pages 20:1--20:16, 2017.
\newblock URL: \url{https://doi.org/10.4230/LIPIcs.ICDT.2017.20}, \href
  {http://dx.doi.org/10.4230/LIPIcs.ICDT.2017.20}
  {\path{doi:10.4230/LIPIcs.ICDT.2017.20}}.

\bibitem{Yannakakis1981}
Mihalis Yannakakis.
\newblock Algorithms for acyclic database schemes.
\newblock In {\em Very Large Data Bases, 7th International Conference,
  September 9-11, 1981, Cannes, France, Proceedings}, pages 82--94, 1981.

\end{thebibliography}

\end{document}